%% file: main.tex
\documentclass[11pt]{article}

\usepackage{bbm}
\usepackage{amsfonts}
\usepackage{amsmath}
\usepackage{amssymb}
\usepackage{amsthm}
\usepackage[mathscr]{eucal}
\usepackage{bbold}
\usepackage{braket}
\usepackage{color}
\usepackage{dsfont}
\usepackage{framed}
\usepackage{jheppub}
\usepackage{mathtools}
\usepackage{physics}
\usepackage[normalem]{ulem}
\usepackage{tensor}
\usepackage{thmtools}
\usepackage{thm-restate}
\usepackage{ulem}
\usepackage{tikz}
\usepackage{soul}

\usetikzlibrary{math} 

\definecolor{acolor}{RGB}{0,192,0}

\definecolor{gcolor}{RGB}{255,0,0}

\definecolor{rcolor}{RGB}{255,0,0}

\definecolor{bcolor}{RGB}{10,10,255}

\newcommand{\ignore}[1]{}

\newtheorem*{eg}{Example}

\newtheorem{nthm}{Theorem}[]
\newtheorem{nlemma}{Lemma}

\declaretheorem[sibling=nthm]{theorem}

\renewcommand{\vec}[1]{\boldsymbol{\mathbf{#1}}}
\renewcommand{\(}{\left(}
\renewcommand{\)}{\right)}
\renewcommand{\[}{\left[}
\renewcommand{\]}{\right]}

\newcommand{\Z}{\mathbb{Z}}

\makeatletter\def\@fpheader{~}\makeatother
\newcommand{\Eqref}[1]{Eq.~\eqref{#1}}
\newcommand{\secref}[1]{Sec.~\ref{#1}}
\newcommand{\figref}[1]{Fig.~\ref{#1}}
 
\usepackage[export]{adjustbox}
 \usepackage{simpler-wick}

  \let\g=\gamma \let\d=\delta 
       \let\p=\phi \let\r=v

\let\y=\psi

\let\la=\label  \let\re=\ref

\let\pa=\partial

\def\be{\begin{equation}}
\def\ee{\end{equation}}
\def\ba#1\ea{\begin{align}#1\end{align}}
\def\bg#1\eg{\begin{gather}#1\end{gather}}
\def\bm#1\em{\begin{multline}#1\end{multline}}
\def\bmd#1\emd{\begin{multlined}#1\end{multlined}}

\def\d{\delta}

\def\g{\gamma}

\def\p{\phi}

\def\r{\rho}

\def\y{\psi}

\def\la{\label}

\def\re{\ref}
\def\er{\Eqref}

\def\fr{\frac}

\def\pa{\partial}

\def\ap{\approx}

\def\qu{\quad}
\def\qqu{\qquad}
\def\lt{\left}
\def\rt{\right}
\def\({\left(}
\def\){\right)}
\def\[{\left[}
\def\]{\right]}
\def\<{\langle}
\def\>{\rangle}

\def\Tr{\operatorname{Tr}}

\usepackage[export]{adjustbox}
 \usepackage{simpler-wick}

\title{Entanglement Negativity and Replica Symmetry Breaking in General Holographic States}

\author[1]{Xi Dong,}
\emailAdd{xidong@ucsb.edu}

\author[2,3]{Jonah Kudler-Flam,}
\emailAdd{jkudlerflam@ias.edu}

\author[4]{and Pratik Rath}
\emailAdd{pratik\_rath@berkeley.edu}

\affiliation[1]{Department of Physics, University of California, Santa Barbara, CA 93106, USA}

\affiliation[2]{School of Natural Sciences, Institute for Advanced Study, Princeton, NJ 08540, USA}
\affiliation[3]{Princeton Center for Theoretical Science, Princeton University, Princeton, NJ~08540, USA}
\affiliation[4]{Center for Theoretical Physics and Department of Physics,
University of California, Berkeley, CA 94720, USA}

\abstract{The entanglement negativity $\mathcal{E}(A:B)$ is a useful measure of quantum entanglement in bipartite mixed states.
In random tensor networks (RTNs), which are related to fixed-area states, it was found in Ref.~\cite{2021JHEP...06..024D} that the dominant saddles computing the even R\'enyi negativity $\mathcal{E}^{(2k)}$ generically break the $\mathbb{Z}_{2k}$ replica symmetry.
This calls into question previous calculations of holographic negativity using 2D CFT techniques that assumed $\mathbb{Z}_{2k}$ replica symmetry and proposed that the negativity was related to the entanglement wedge cross section.
In this paper, we resolve this issue by showing that in general holographic states, the saddles computing $\mathcal{E}^{(2k)}$ indeed break the $\mathbb{Z}_{2k}$ replica symmetry.

Our argument involves an identity relating $\mathcal{E}^{(2k)}$ to the $k$-th R\'enyi entropy on subregion $AB^*$ in the doubled state $\ket{\rho_{AB}}_{AA^*BB^*}$, from which we see that the $\mathbb{Z}_{2k}$ replica symmetry is broken down to $\mathbb{Z}_{k}$. 
For $k<1$, which includes the case of $\mathcal{E}(A:B)$ at $k=1/2$, we use a modified cosmic brane proposal to derive a new holographic prescription for $\mathcal{E}^{(2k)}$ and show that it is given by a new saddle with multiple cosmic branes anchored to subregions $A$ and $B$ in the original state.
Using our prescription, we reproduce known results for the PSSY model and show that our saddle dominates over previously proposed CFT calculations near $k=1$.
Moreover, we argue that the $\mathbb{Z}_{2k}$ symmetric configurations previously proposed are not gravitational saddles, unlike our proposal.
Finally, we contrast holographic calculations with those arising from RTNs with non-maximally entangled links, demonstrating that the qualitative form of backreaction in such RTNs is different from that in gravity.}

\begin{document}
 

\maketitle

\section{Introduction}\label{sec:intro}

It is well known that entanglement plays a crucial role in the emergence of a semiclassical spacetime description in holographic systems \cite{2010GReGr..42.2323V,2013ForPh..61..781M, 2015JHEP...04..163A, 2016PhRvL.117b1601D}. 
While this connection has been well understood in the context of calculating entanglement entropy via the Ryu-Takayanagi (RT) formula and its subsequent generalizations \cite{2006PhRvL..96r1602R,2007JHEP...07..062H,2015JHEP...01..073E}, the precise structure of multipartite entanglement has only recently been explored  \cite{2019arXiv190500577D, 2019PhRvD..99j6014K,2019PhRvL.123m1603K,2021JHEP...06..024D, 2018NatPh..14..573U,2020JHEP...04..208A,2018JHEP...01..098N,2019PhRvL.122n1601T,2018JHEP...10..152U}. 
For general mixed states $\rho_{AB}$, the correlations between subsystems $A$ and $B$ can be both classical and quantum. 
A state with only classical correlations is called separable and takes the form
\begin{align}
    \rho_{AB} = \sum_i p_i \rho_A^{(i)} \otimes \rho_B^{(i)}, \qquad p_i\geq 0, \qquad \sum_i p_i = 1,
\end{align}
where $\rho_{A/B}^{(i)}$ are density matrices themselves.
The entanglement entropy (or more precisely the mutual information) is sensitive to both classical and quantum correlations between $A$ and $B$, and generally does not vanish for separable states. 
Thus, it is of clear interest to obtain a quantity that faithfully measures only quantum entanglement.
Among a whole zoo of such measures,\footnote{See, for instance, Ref.~\cite{Plenio:2007zz} for a review of such entanglement measures.} we are interested in a measure that is computable, operationally meaningful, and potentially has a geometric interpretation in holography. 

\subsubsection*{Entanglement negativity}
With this motivation, our primary focus in this paper will be on the logarithmic negativity (henceforth called the ``negativity''), which is a genuine measure of entanglement \cite{1998PhRvA..58..883Z,2002PhRvA..65c2314V,1996PhRvL..77.1413P,1999JMOp...46..145E,2005PhRvL..95i0503P,2000PhRvL..84.2726S,1996PhLA..223....1H} because it vanishes for all separable states and decreases monotonically under local operations and classical communications \cite{2005PhRvL..95i0503P}. 
Negativity is defined as
        \begin{equation}
                \mathcal{E}(A:B) = \log \left(\sum_{i} |\lambda_i^{(T)}|\right),
        \end{equation}
where $\lambda_i^{(T)}$ are the eigenvalues of $\rho^{T_B}_{AB}$, obtained by performing a partial transpose $T_B$ of the density matrix $\rho_{AB}$.
In order to analyze the negativity, it is useful to study a family of quantities called the even R\'enyi negativity (ERN)\footnote{There is a separate family of R\'enyi negativities arising from the odd moments of $\rho^{T_B}_{AB}$. 
In this paper, we primarily focus on the even case which is related to the negativity $\mathcal{E}$. 
We comment on the odd case in \secref{sub:odd}.}
\begin{align}
    \mathcal{E}^{(2k)}(A:B) &= \log\( \sum_{i} |\lambda_i^{(T)}|^{2k}\),\end{align}
which can be analyzed using a replica trick since $\sum_{i} |\lambda_i^{(T)}|^{2k}=\Tr \[\(\rho_{AB}^{T_B} \)^{2k}\]$ for integer $k$. Notably, the replica trick for ERN has a $\mathbb{Z}_{2k}$ symmetry that cyclically permutes the replicas. This quantity can be analytically continued to other values of $k$.
The negativity is then given by the $k=1/2$ ERN.

\begin{figure}
    \centering
    \includegraphics[scale=0.5]{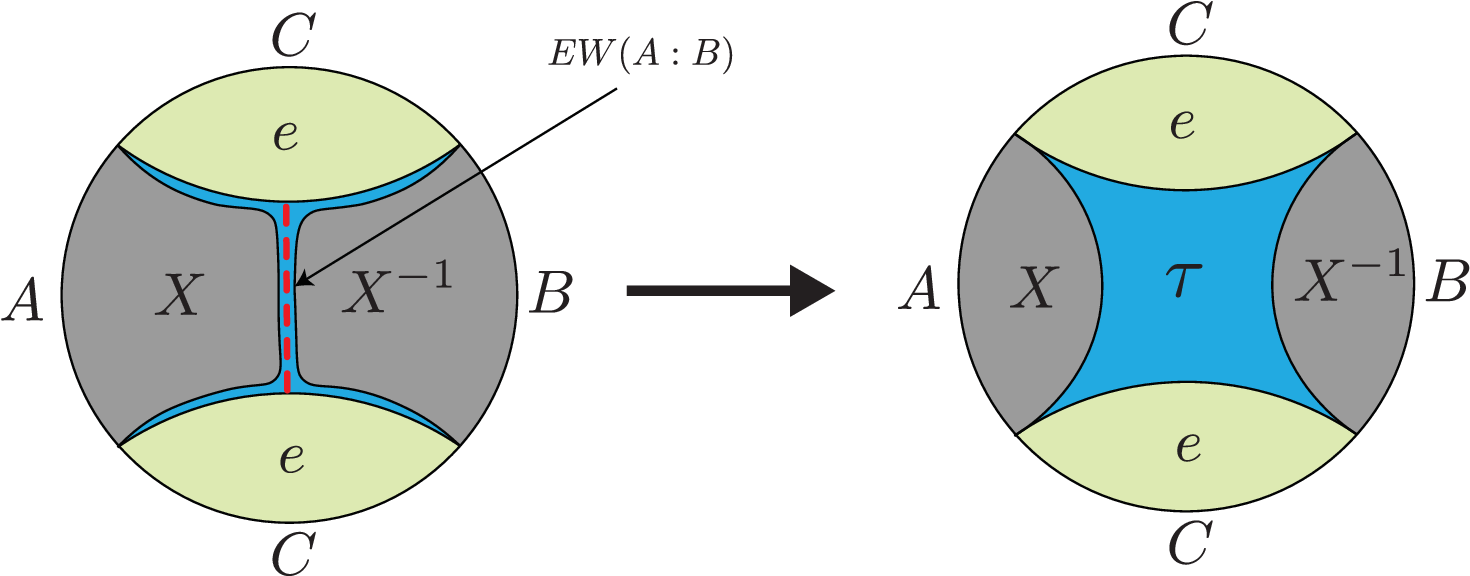}
    \caption{
    The RTN calculation of even moments of the partially transposed density matrix $\rho_{AB}^{T_B}$ is given by the free energy of a spin model with boundary conditions $\{X,X^{-1},e\}$ at the subregions $\{A,B,C\}$ respectively. $\{X,X^{-1},e\}$ denote the cyclic, anti-cyclic and identity permutations respectively. In the connected phase, the naive saddle (left) has a domain wall at the entanglement wedge cross section $EW(A:B)$ (red, dashed), which is the minimal area surface that divides the entanglement wedge of $AB$, bounded by the RT surface $\gamma_{C}$, into portions containing subregions $A$ and $B$ respectively. The naive saddle can be split (left) by a small domain of an at-most-$\mathbb{Z}_k$-symmetric permutation $\tau$ (blue) and eventually allowed to relax to the true ground state configuration (right).}
    \label{splitting_fig}
\end{figure}

Given the usefulness of the negativity, it is natural to ask whether it has a useful holographic dual. 
Several previous attempts have been made to answer this question \cite{2014JHEP...10..060R, 2016arXiv160906609C, 2019PhRvD..99j6014K,2021JHEP...06..024D}. However, no consensus has been reached on a universal bulk dual. 
We will briefly review some of the literature on this topic below.

\subsubsection*{Holographic negativity}

Negativity was computed in 2D holographic conformal field theory (CFT) in Refs.~\cite{2014JHEP...09..010K, 2019PhRvD..99j6014K,2019PhRvL.123m1603K}. 
This calculation was done under the assumption of the dominance of a single Virasoro conformal block, which is standard in computations of entanglement entropy in 2D holographic CFTs \cite{2013arXiv1303.6955H}. 
On the bulk side, this translated into a $\mathbb{Z}_{2k}$ replica symmetric configuration for $\mathcal{E}^{(2k)}$, resulting in $\mathcal{E}(A:B)$ being related to a backreacted version of the entanglement wedge cross section, $EW(A:B)$, shown in \figref{splitting_fig}. 

Meanwhile, this problem was also analyzed in random tensor networks (RTNs) in Ref.~\cite{2021JHEP...06..024D}. 
RTNs are useful toy models for holography and accurately model fixed-area states in gravity that possess a flat entanglement spectrum \cite{2019JHEP...10..240D,2020JHEP...03..191D,2019JHEP...05..052A}. 
RTNs allow for a rigorous calculation of the negativity as well as the ERN.
Somewhat surprisingly, the dominant saddle, shown on the right side of \figref{splitting_fig}, has at most $\mathbb{Z}_k$ symmetry generically.
This leads to $\mathcal{E}(A:B) = \frac{1}{2}I(A:B) + O(1)$, where $I(A:B)$ is the mutual information and $O(1)$ denotes subleading (in Newton's constant $G$) corrections. In particular, the mutual information is computed by extremal surfaces anchored to the subregions, quite unlike $EW(A:B)$.
More generally, assuming a $\mathbb{Z}_k$ replica symmetry motivated by the RTN calculation, a prescription was provided in Ref.~\cite{2021JHEP...06..024D} to compute the negativity for general holographic states using the cosmic brane proposal\footnote{We will henceforth refer to it as the original cosmic brane proposal to distinguish it from the modified cosmic brane proposal of Ref.~\cite{Dong:2023bfy}.} of Refs.~\cite{2013JHEP...08..090L,2016NatCo...712472D}.

\subsubsection*{Our results}

The goal of this paper is to revisit the proposal of Ref.~\cite{2021JHEP...06..024D} for general holographic states with arbitrary entanglement spectra. 
Our first result in this paper is to demonstrate that in general, only a $\mathbb{Z}_k$ replica symmetry is preserved by the saddles computing $\mathcal{E}^{(2k)}$ at integer $k$ for general holographic states.
This allows us to obtain our second result, an explicit holographic proposal for the ERN at arbitrary $k$, and thus also for the negativity.
The final answer, which we shall briefly summarize below, agrees with the prescription for computing ERN provided by Ref.~\cite{2021JHEP...06..024D} only for $k\geq 1$.
On the other hand, for $k<1$ the original cosmic brane proposal provably fails at leading order, and we instead need to apply the modified cosmic brane proposal provided in Ref.~\cite{Dong:2023bfy}.
In particular, this modified cosmic brane proposal is crucial to reproduce the negativity itself since it arises as the $k=1/2$ ERN.

\begin{figure}
    \centering
    \includegraphics[scale=0.6]{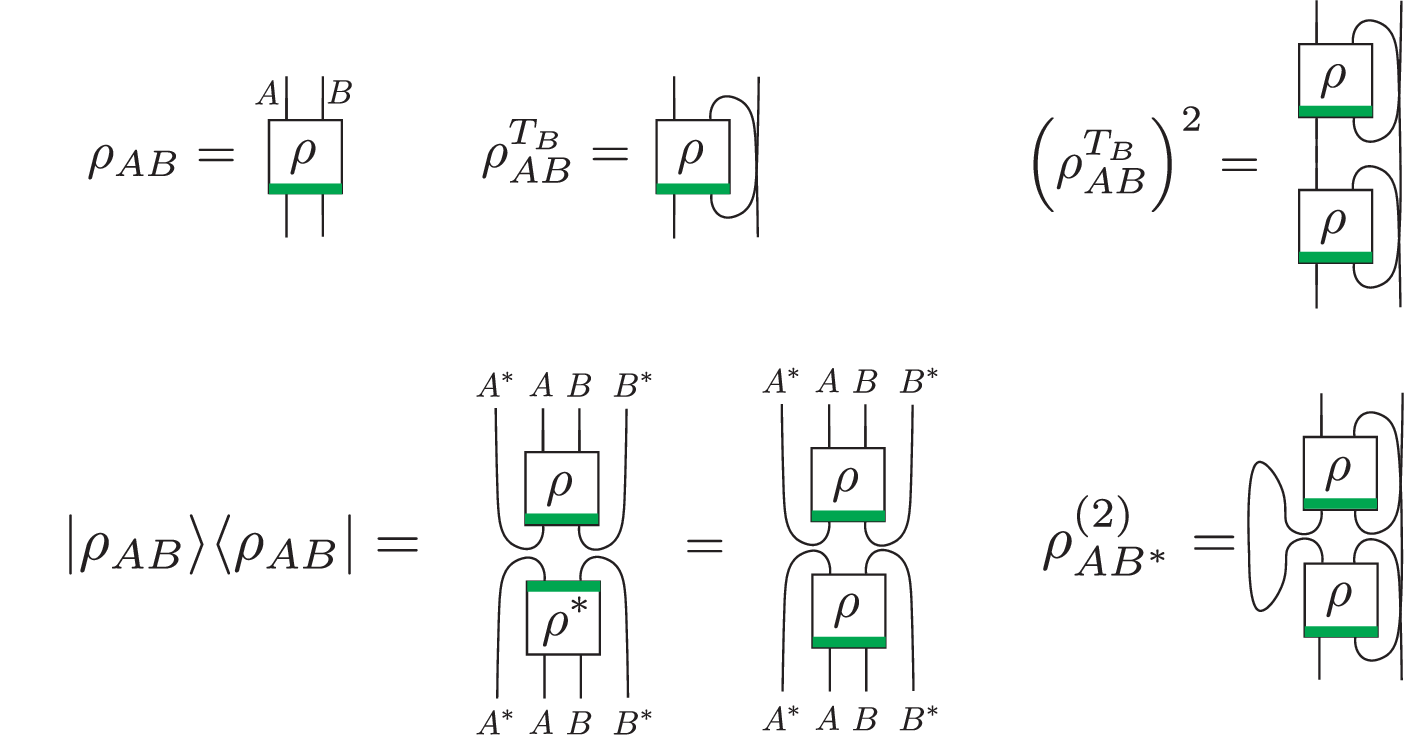}
    \caption{A diagrammatic proof of the identity \Eqref{eq:central_identity} where the boxes represent tensors, a transpose is performed by switching legs and tensor contraction is done by gluing legs together. We have added a green strip to clarify the orientation of the tensor. In the second line, we first obtain $\bra{\rho_{AB}}$ by a mirror image and then use hermiticity of $\rho_{AB}$ in the second step.}
    \label{fig:identity}
\end{figure}

In order to obtain our results, we find a very useful identity relating the partially transposed density matrix $\rho_{AB}^{T_B}$ to the doubled (``Choi'') state $\ket{\rho_{AB}}_{AA^*BB^*}$, which is obtained by applying channel-state duality to the operator $\rho_{AB}$ to convert it into a pure state in a doubled Hilbert space $\mathcal{H}_A\otimes\mathcal{H}_{A^*}\otimes\mathcal{H}_B\otimes\mathcal{H}_{B^*}$ \cite{choi1975completely,jamiolkowski1972linear}.
In particular, we have the identity
\begin{equation}
\label{eq:central_identity}
    \left(\rho_{AB}^{T_B}\right)^2 = \rho^{(2)}_{AB^*}\equiv\Tr_{A^* B}\[\ket{\rho_{AB}}\bra{\rho_{AB}}\],
\end{equation}
where $\ket{\rho_{AB}}$ and $\rho^{(2)}_{AB^*}$ are unnormalized.\footnote{The state $\ket{\rho_{AB}}$ belongs to the one-parameter generalization $\ket{\rho_{AB}^{m/2}}$ which includes the canonical purification $\ket{\sqrt{\rho_{AB}}}$.
Such states have been well studied in the context of the reflected entropy whose holographic dual is related to $EW(A:B)$ \cite{2019arXiv190500577D}.}
We provide a diagrammatic proof of this fact in \figref{fig:identity}.

The $2k$-th moments of the partially transposed density matrix are thus related to the $k$-th moments of a properly normalized density matrix defined as $\bar{\rho}^{(2)}_{AB^*}=\frac{\rho^{(2)}_{AB^*}}{\Tr \left(\rho^{(2)}_{AB^*}\right)}$:
\begin{align}
 \Tr\[\left(\rho_{AB}^{T_B}\right)^{2k}\] = \Tr \[\left(\rho^{(2)}_{AB^*}\)^k\] = \Tr \[\(\bar{\rho}^{(2)}_{AB^*}\)^k\] \(\Tr \rho^{(2)}_{AB^*}\)^k.
\end{align}
The ERN then becomes
\begin{equation}
    \mathcal{E}^{(2k)}(A:B) = -(k-1)S_k\(\bar{\rho}^{(2)}_{AB^*}\) - k S_2\(\rho_{AB}\),
\end{equation}
where $S_n(\rho)$ is the R\'enyi entropy of $\rho$.

\begin{figure}
    \centering
    \includegraphics[scale=0.5]{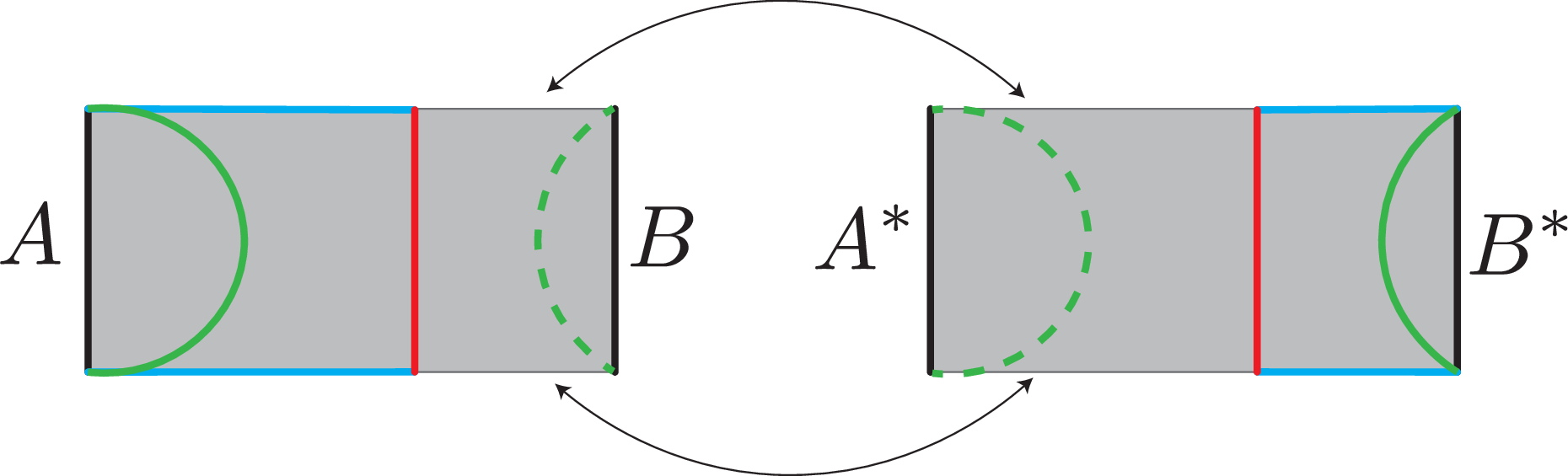}
    \caption{The spatial geometry dual to the doubled state $\ket{\rho_{AB}}_{AA^*BB^*}$ for $A,B$ chosen to be disjoint intervals in the vacuum state of a 2D holographic CFT is depicted. The left and right horizontal edges are identified so that the spatial slice is topologically a cylinder. The saddles computing R\'enyi entropy for $AB^*$ involve cosmic branes anchored to the respective subregions. For $k\geq1$, the true location of the cosmic brane of tension $T_k$ (opening angle $\frac{2\pi}{k}$) is at either the solid or dashed green surface, thus breaking the $\mathbb{Z}_2$ symmetry exchanging $AB$ with $A^*B^*$. The previously proposed $\mathbb{Z}_2$ symmetric configuration involves a union of cosmic branes on the entanglement wedge cross section (red) and the RT surface for $AB$ (blue; note that the left and right horizontal edges are identified). For $k<1$, the dominant saddle has cosmic branes of tension $\frac{T_k}{2}$ (opening angle $\pi+\frac{\pi}{k}$) at both the solid and dashed green surfaces.}
    \label{fig:cylinder}
\end{figure}

The calculation of R\'enyi entropy is by now standard in the context of holography, and there is a lot of evidence that at integer $k$, the dominant saddles are $\mathbb{Z}_k$ replica symmetric \cite{2013JHEP...08..090L,2016NatCo...712472D}.
Using this assumption, one can then quotient the bulk saddle by the $\mathbb{Z}_k$ replica symmetry to obtain a manifold with conical defects located at the fixed points of the $\mathbb{Z}_k$ symmetry.
Such conical defects can be thought of as being sourced by cosmic branes of tension $T_k=\frac{k-1}{4kG}$.

The preservation or breaking of $\mathbb{Z}_{2k}$ symmetry in the bulk saddles for $ \Tr\left(\rho_{AB}^{T_B}\right)^{2k} $ is then controlled by the location of the cosmic branes and whether they preserve the remaining $\mathbb{Z}_2$ symmetry.
For integer $k>1$, it is then easy to see in examples that this $\mathbb{Z}_2$ symmetry is generally broken. 
We demonstrate one such example in \figref{fig:cylinder} which arises in the computation of negativity for two disjoint intervals in the vacuum state. 
In the phase where the entanglement wedge of the union of the two intervals is connected, preservation of the $\mathbb{Z}_2$ symmetry requires the cosmic branes to intersect orthogonally, but such a configuration cannot descend from a $\mathbb{Z}_k$ quotient of a smooth parent manifold, as we prove in Appendix \ref{app:brane}. 
Moreover, there is a natural continuation of the cosmic brane saddles to arbitrary $k\geq1$.
Near $k=1$, where the cosmic brane becomes a probe RT surface, it is easy to see that any ``saddle'' with intersecting branes, even if included, would be subleading, since the area can be decreased by smoothing the corners. 
In this example, the $\mathbb{Z}_{2k}$ symmetric configuration has cosmic branes at the connected RT surface $\gamma_C$, $EW(A:B)$, and $EW(A^*:B^*)$. 
On the other hand, there are two $\mathbb{Z}_k$ symmetric configurations involving cosmic branes of tension $T_k$ at either $\gamma_{A}\cup \gamma_{B^*}$ or $\gamma_{A^*}\cup \gamma_{B}$, which are clearly smaller in area compared to the $\mathbb{Z}_{2k}$ symmetric configuration.
This breaking of the $\mathbb{Z}_{2k}$ replica symmetry in general holographic states is our first result.

A key feature of these $\mathbb{Z}_k$ symmetric saddles is that they are exactly degenerate with their image under the $\mathbb{Z}_2$ transformation exchanging $AB\leftrightarrow A^*B^*$. 
This means that for the ERN, we are always situated precisely at a ``phase transition.'' 
In general, and particularly near phase transitions, there can be leading order corrections to the R\'enyi entropies for $k<1$ \cite{2019arXiv191111977P,2020JHEP...11..007D,2020JHEP...12..084M,Dong:2023bfy}.
In this situation, we instead need to apply the modified cosmic brane proposal of Ref.~\cite{Dong:2023bfy} to compute the R\'enyi entropy.

In particular, for the ERN, we will demonstrate that the original cosmic brane proposal fails for all $k<1$, except in a few special situations.
Moreover, we prove that the ERN for $k<1$ is dominated by a saddle in the ``diagonal phase," where two cosmic branes with equal areas are on surfaces $\gamma_{A}\cup \gamma_{B^*}$ and $\gamma_{A^*}\cup \gamma_{B}$; see e.g., \figref{fig:cylinder}.
Interestingly, this restores the $\mathbb{Z}_2$ symmetry that was lost for $k>1$, allowing us to perform a further $\mathbb{Z}_2$ quotient.
This leads to our second result: a new, concrete holographic proposal for the ERN in the connected phase at arbitrary $k<1$ summarized by
\begin{equation}\la{ernresult}
    \mathcal{E}^{(2k)}(A:B)= 2 k \[I\(M_1\) - I\(M_1,\g_{AB}^{(\pi)},\(\g_{A}\cup\g_B\)^{(\pi+\pi/k)}\)\],
\end{equation}
where $I(f)$ is the gravitational action of the solution with boundary conditions $f$, $M_1$ represents the boundary conditions preparing the original state, and $\g_i^{(\phi)}$ represents the insertion of a conical defect of opening angle $\phi$ at surface $\g_i$.\footnote{In fixed-area states, we use the more general definition that $\g_i^{(\phi)}$ represents the insertion of a cosmic brane of tension $T_{2\pi/\p} = \frac{2\pi-\p}{8\pi G}$. Under this general definition, Eqs.~(\re{ernresult}) and (\re{lnresult}) also hold in the disconnected phase, where $\g_{AB}^{(\pi)} = \(\g_{A}\cup\g_B\)^{(\pi)}$ coincides and thus partially cancels $\(\g_{A}\cup\g_B\)^{(\pi+\pi/k)}$, resulting in a cosmic brane of net tension $\frac{2k-1}{8kG}$.}
In particular, the holographic proposal for the negativity in the connected phase is
\begin{align}\la{lnresult}
    \mathcal{E}(A:B) = I\(M_1\) - I\(M_1,\g_{AB}^{(\pi)},\(\g_{A}\cup\g_B\)^{(3\pi)}\).
\end{align}

\subsubsection*{Overview}

In \secref{sec:gravity}, we discuss the holographic dual of negativity.
Motivated by the identity \Eqref{eq:central_identity}, we formulate the holographic dual in terms of the doubled state $\ket{\rho_{AB}}$.
We first review the holographic construction of $\ket{\rho_{AB}}$ using the gravitational path integral.
Using this, we compute the ERN by applying the modified cosmic brane proposal to subregion $AB^*$ in the $\ket{\rho_{AB}}$ state.
We show that the original cosmic brane proposal must generally fail for any $k<1$ (except for very special cases).
Moreover, we show that the ERN for $k<1$ is always dominated by a saddle in the diagonal phase, resulting in a simple bulk dual for the negativity.

In \secref{sec:PSSY}, we illustrate our proposal using the simple example of the PSSY model \cite{2019arXiv191111977P}, reproducing the results obtained in Ref.~\cite{2021arXiv211011947D} and finding some new results.

In \secref{sec:intervals}, we revisit the example of computing the negativity for two disjoint intervals in the vacuum state of a 2D holographic CFT.  
We describe how the wrong assumption of $\mathbb{Z}_{2k}$ symmetry and dominance of a single channel was used in the calculations in Refs.~\cite{2019PhRvD..99j6014K,2019PhRvL.123m1603K}. 
In Appendix~\ref{sub:PRMI}, we provide a simpler example of the Petz R\'enyi mutual information, where a calculation under analogous assumptions can be performed that leads to an obviously incorrect answer.
We argue that the $\mathbb{Z}_{2k}$ symmetric configurations proposed by Refs.~\cite{2019PhRvD..99j6014K,2019PhRvL.123m1603K} fail to be gravitational saddles and, moreover, we demonstrate that our saddle dominates over their non-saddle contribution for the calculation of $\tilde{\mathcal{E}}^{(2)}$.
This provides significant evidence for our argument.

In \secref{sec:disc}, we summarize our results and discuss various aspects of our work. 
In \secref{sub:odd}, we discuss the calculation of odd moments of the partially transposed density matrix. 
In \secref{sub:nmRTNdisc}, we discuss shortcomings of random tensor networks with non-flat entanglement spectra (nfRTNs) and potential ways to improve them as models of holography.
The detailed differences in nfRTNs and gravity are explained in Appendix~\ref{sec:RTN}.
Appendix~\ref{app:brane} discusses brane intersections that can descend from a quotient of smooth parent spacetimes.

\section{Holographic Dual of Negativity}
\label{sec:gravity}

In this section, we describe our proposal for the holographic dual of negativity.
First, we describe the gravity dual of the auxiliary state $\ket{\rho_{AB}}$. 
With this state in hand, we simply need to evaluate the R\'enyi entropies. 
We review the modified cosmic brane proposal for computing R\'enyi entropy holographically.
Putting these ingredients together, we arrive at a proposal for the holographic dual of ERNs, and most importantly the logarithmic negativity itself.

\subsection{The holographic dual of $\ket{\rho_{AB}}$}

For simplicity, let us consider a CFT state $\ket{\psi}_{ABC}$ that enjoys a time-reversal symmetry and can be prepared using a Euclidean path integral on a manifold $M_1$.
By the AdS/CFT dictionary, its bulk dual can then be obtained from the corresponding gravitational saddle consistent with the boundary conditions specified by the CFT path integral.
In particular, we use the Euclidean path integral to compute its norm $\bra{\psi}\ket{\psi}$ and the corresponding Euclidean bulk geometry is labelled $\mathcal{B}_1$ which satisfies $\pa \mathcal{B}_1=M_1$.
The $\mathbb{Z}_2$ symmetric slice $\Sigma_1$ of $\mathcal{B}_1$ provides initial data for obtaining the Lorentzian spacetime associated to $\ket{\psi}$.

Given the reduced density matrix $\rho_{AB}$ on subregion $AB$ in state $\ket{\psi}$, we would like to reinterpret it as a pure state $\ket{\rho_{AB}}$ that lives in a doubled Hilbert space $\mathcal{H}_{AB}\otimes \mathcal{H}_{A^*B^*}$.\footnote{This procedure of mapping operators acting on $\mathcal{H}_{AB}$ to states in a doubled Hilbert space is familiar in the context of canonical purification, and the interested reader can find more details in Ref.~\cite{2019arXiv190500577D}. In mathematical literature, it is commonly known as the Choi–Jamiolkowski isomorphism \cite{choi1975completely,jamiolkowski1972linear}.
}
In particular, the map defines the inner product between states as $\bra{\sigma_1}\ket{\sigma_2}=\Tr\[\sigma_1^\dagger\sigma_2\]$, where $\sigma_1,\sigma_2$ are operators acting on $\mathcal{H}_{AB}$.

More explicitly, let us pick a basis $\ket{i}_A,\ket{j}_B$ for $\mathcal{H}_A,\mathcal{H}_B$ respectively.
In this basis, the reduced density matrix $\rho_{AB}$ takes the form
\begin{equation}
    \rho_{AB} = \sum_{i,i'=1}^{d_A} \sum_{j,j'=1}^{d_B} \rho_{iji'j'}\ket{i}_A\ket{j}_B\bra{i'}_A\bra{j'}_B.
\end{equation}
Then the doubled state $\ket{\rho_{AB}}$ is given by
\begin{equation}
    \ket{\rho_{AB}} = \sum_{i,i'=1}^{d_A} \sum_{j,j'=1}^{d_B} \rho_{iji'j'}\ket{i}_A\ket{j}_B\ket{i'}_{A^*}\ket{j'}_{B^*}.
\end{equation}
Using these expressions, it is also straightforward to derive the identity \Eqref{eq:central_identity} for which we have provided a diagrammatic proof in \figref{fig:identity}.
Note that the state $\ket{\rho_{AB}}$ depends on the choice of the basis $\ket{i}_A,\ket{j}_B$ for $\mathcal{H}_A,\mathcal{H}_B$ used in transforming to the doubled Hilbert space.
However, entropies computed in $\ket{\rho_{AB}}$ for any combination of $A,A^*,B,B^*$ are independent of the basis choice for $\mathcal{H}_A,\mathcal{H}_B$.

\begin{figure}
    \centering
    \includegraphics[scale=0.25]{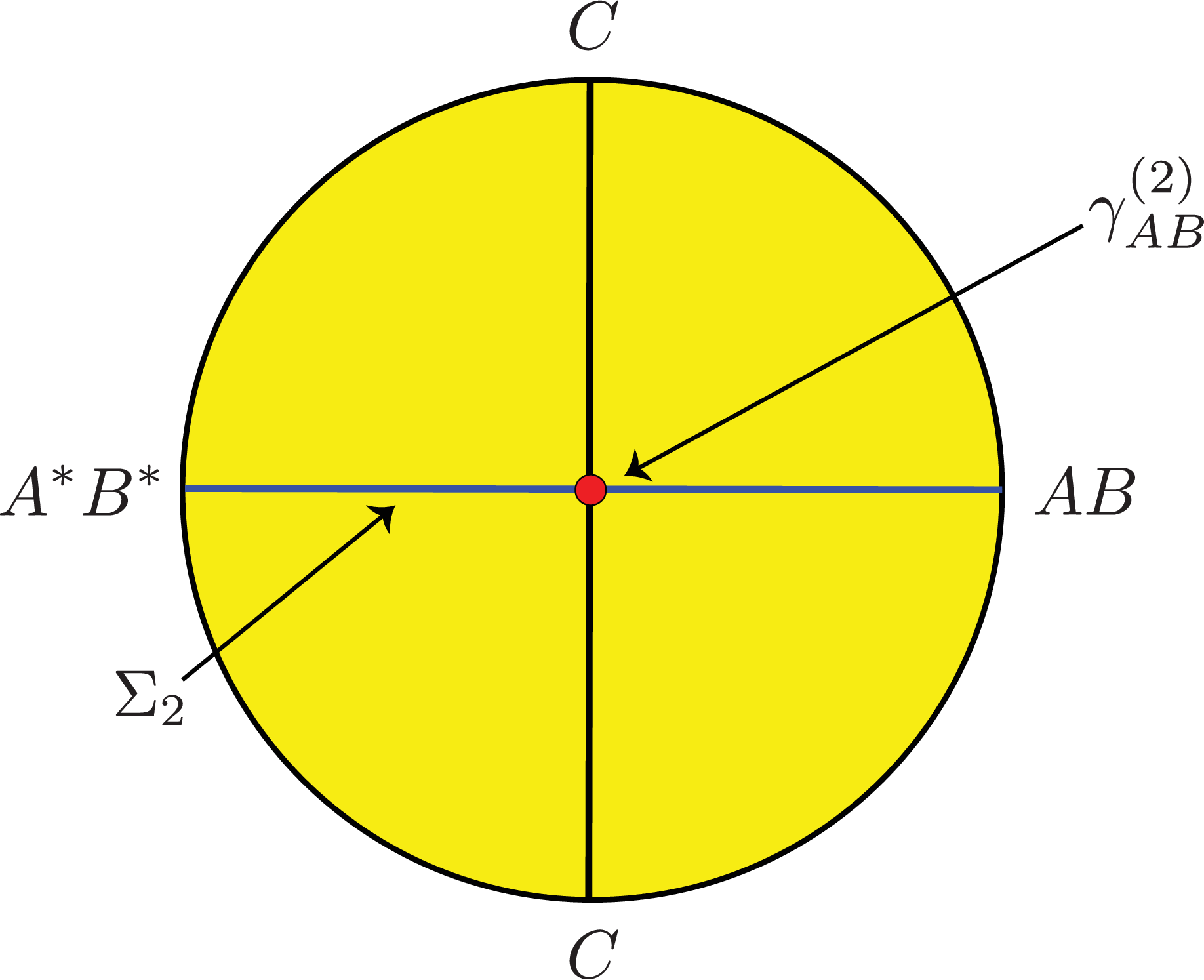}
    \caption{The Euclidean geometry $\mathcal{B}_2$ computing the norm of $\ket{\rho_{AB}}$ has a $\mathbb{Z}_2$ time reversal symmetry, as well as a $\mathbb{Z}_2$ symmetry under the exchange $AB\leftrightarrow A^* B^*$. The time reversal symmetric slice $\Sigma_2$ (blue) contains $\gamma_{AB}^{(2)}$ (red) which is the RT surface for subregion $AB$ in the state $\ket{\rho_{AB}}$.}
    \label{fig:doubled}
\end{figure}

The norm of the doubled state is given by\footnote{We remind the reader that $\ket{\rho_{AB}}$ is unnormalized in general.}
\begin{equation}\label{eq:norm}
    \bra{\rho_{AB}}\ket{\rho_{AB}} = \Tr\[\rho_{AB}^2\],
\end{equation}
which can be computed using a Euclidean path integral in the CFT on manifold $M_2^{AB}$ that is a double branched cover of $M_1$ over subregion $AB$.
Following the same logic as above, we can find the dual bulk geometry associated to $\ket{\rho_{AB}}$ by considering the Euclidean saddle $\mathcal{B}_2$ such that $\pa \mathcal{B}_2=M_2^{AB}$.
Following Ref.~\cite{2013JHEP...08..090L}, it is standard to assume that the dominant saddle $\mathcal{B}_2$ respects the $\mathbb{Z}_2$ symmetry permuting the two copies of the boundary CFT glued together in computing \Eqref{eq:norm}.
Moreover, in the presence of a time-reflection symmetry as we have assumed, this symmetry is enhanced to a dihedral $D_2$ symmetry \cite{2019arXiv190500577D}.

Cutting open $\mathcal{B}_2$ on the time-symmetric slice $\Sigma_2$ gives us the initial data for obtaining the Lorentzian spacetime associated to $\ket{\rho_{AB}}$.
For the purpose of computing $\mathcal{E}^{(2k)}(A:B)$, we are interested in computing the $k$-th R\'enyi entropy of subregion $AB^*$ in the state $\ket{\rho_{AB}}$.
In order to do so, we will briefly review the computation of R\'enyi entropies in general holographic states closely following Ref.~\cite{Dong:2023bfy}.

\subsection{Holographic R\'enyi entropy}
In a state $\ket{\psi}$, the R\'enyi entropy of a subregion $R$ is defined as
\begin{align}
    S_n(R) = \frac{1}{1-n}\log \Tr\rho_R^n,
\end{align}
where $\rho_R$ is the reduced density matrix on $R$.
For integer $n>1$, the R\'enyi entropy may be computed via the gravitational path integral using the standard replica trick.
In the semi-classical limit ($G \rightarrow 0$), the saddle point approximation is valid and we may approximate the path integral by a single gravitational configuration, $\mathcal{B}_n$. 
It is standard to assume that this dominant configuration in the bulk preserves the $\mathbb{Z}_n$ replica symmetry of the boundary manifold $M_n^R$ \cite{2013JHEP...08..090L} which is an $n$-sheeted branched cover over $R$ of manifold $M_1$ that computes the norm $\braket{\psi}$.
By now there is a lot of evidence in the literature for this assumption of replica symmetry to be true at leading order in $G$.

Assuming replica symmetry, one may then quotient $\mathcal{B}_n$ by $\mathbb{Z}_n$ to obtain a new geometry $\hat{\mathcal{B}}_n$ that has a conical defect with opening angle $\frac{2\pi}{n}$ emanating from the branching points of $M_n^R$. 
While $\mathcal{B}_n$ made sense only for integer $n$, there is a natural continuation of $\hat{\mathcal{B}}_n$ to non-integer $n$ by tuning the opening angle. 
This is equivalent to solving the bulk equations of motion with a cosmic brane anchored to the boundary entanglement surface with tension $T_n = \frac{n-1}{4nG}$ \cite{2013JHEP...08..090L,2016NatCo...712472D}. 
The moments of the density matrix are then
\begin{align}
    \Tr \rho_R^n = e^{-n(I[\hat{\mathcal{B}}_n]-I[\hat{\mathcal{B}}_1])},
\end{align}
where $I[\mathcal{B}]$ is the gravitational action of bulk manifold $\mathcal{B}$. 

In Ref.~\cite{Dong:2023bfy}, we have argued that the original cosmic brane proposal can fail when there are two or more candidate extremal surfaces for the subregion of interest, in which case one must employ a modified cosmic brane proposal that correctly computes the R\'enyi entropy even for $n<1$.
Consider a situation when there are two candidate extremal surfaces $\gamma_{1,2}$ with areas $A_{1,2}$ as will be relevant for the negativity calculation.
Let $p\(A_1,A_2\)$ be the probability distribution over the two areas in state $\ket{\psi}$ defined as
\begin{align}\label{eq:prob}
       p\(A_1,A_2\)&=\frac{\<\y|P_{A_1,A_2}|\y\>}{\<\y|\y\>},
\end{align}
where $P_{A_1,A_2}$ is a projector onto definite values $A_{1,2}$ of the areas of surfaces $\g_{1,2}$.
The probability distribution can be computed using the gravitational path integral as \cite{2020JHEP...11..007D,2020JHEP...12..084M,Dong:2023bfy,2021JHEP...04..062A}
\begin{align}
       p\(A_1,A_2\)&=\exp\(I\[\mathcal{B}_1\]-I\[\mathcal{B}_{A_1,A_2}\]\),\label{eq:prob2}
\end{align}
where $\mathcal{B}_{A_1,A_2}$ is the fixed-area saddle obtained by solving the equations of motion given $M_1$ as asymptotic boundary condition and areas $A_{1,2}$ at the surfaces $\g_{1,2}$.

The original cosmic brane proposal can then be reformulated as \cite{Dong:2023bfy}
\begin{align}\label{eq:cosmic_A}
    S_n^{C}(R) = \begin{cases}
        \displaystyle \frac{1}{1-n} \max_{i=1,2}\,\max_{A_1,A_2}\( n\log p\(A_1,A_2\) +\frac{(1-n)}{4G}A_i\)&n\geq1,\\
        \displaystyle \frac{1}{1-n} \min_{i=1,2}\,\max_{A_1,A_2}\( n\log p\(A_1,A_2\) +\frac{(1-n)}{4G}A_i\)&n<1.
    \end{cases}
\end{align}
It is useful to understand this by writing down the maximization condition coming from \Eqref{eq:cosmic_A}: e.g., in the case of $i=1$ we have
\begin{align}\label{eq:max11}
    \frac{\pa I\[\mathcal{B}_{A_1,A_2}\]}{\pa A_1}&= \frac{1-n}{4 n G},\\
    \frac{\pa I\[\mathcal{B}_{A_1,A_2}\]}{\pa A_2}&=0,\label{eq:max21}
\end{align}
where we have used \Eqref{eq:prob2}.
These equations are the equations of motion arising from the insertion of a cosmic brane of the appropriate tension $T_n$.
For a holographic state prepared using a smooth gravitational path integral, the right-hand side (RHS) of \Eqref{eq:max11} and \Eqref{eq:max21} can be related to the conical opening angle at the given surface and result in $\phi_1=\frac{2\pi}{n}$ and $\phi_2=2\pi$ as required for the original cosmic brane proposal \cite{Dong:2023bfy}.

On the other hand, the modified cosmic brane proposal, based on the assumption of a diagonal approximation in the fixed-area basis,\footnote{This approximation was later justified in Ref.~\cite{Penington:2024jmt}.} is given by \cite{Dong:2023bfy}
\begin{align}\label{eq:MCS}
    S_n^{MC}(R) = \begin{cases}
        \displaystyle \frac{1}{1-n}\max_{A_1,A_2}\, \max_{i=1,2}\( n\log p\(A_1,A_2\) +\frac{(1-n)}{4G}A_i\)&n\geq1,\\
        \displaystyle \frac{1}{1-n}\max_{A_1,A_2}\, \min_{i=1,2}\( n\log p\(A_1,A_2\) +\frac{(1-n)}{4G}A_i\)&n<1,
    \end{cases}
\end{align}
which differs from the original cosmic brane proposal in the order of optimization. 
It is straightforward to see that the two proposals agree for $n\geq 1$ but can disagree for $n<1$ \cite{Dong:2023bfy}.
For $n<1$, we will compute R\'enyi entropies using the modified cosmic brane proposal \Eqref{eq:MCS}. 

It was shown in Ref.~\cite{Dong:2023bfy} that the modified cosmic brane proposal agrees with the original cosmic brane proposal if and only if one of the original cosmic brane saddles satisfies the minimality constraint, i.e., the area of the cosmic brane is no greater than the area of the other candidate RT surface.
However, for $n<1$, it is possible that neither of the original cosmic brane saddles satisfies the minimality constraint.
In this case, the dominant contribution to the modified cosmic brane proposal comes either from a diagonal saddle with $A_1=A_2$ where the cosmic brane tension is distributed over the two surfaces, or from a subleading saddle for the original cosmic brane proposal which satisfies the minimality constraint.
This will in fact always turn out to be the case for the computation of ERNs with $k<1$, which we now turn to.

\subsection{Holographic negativity}

Using our identity \Eqref{eq:central_identity}, we have the following formula for the ERN as discussed in \secref{sec:intro}:
\begin{align}\label{eq:results}
\mathcal{E}^{(2k)}(A:B) &= -(k-1)S_k\(\bar{\rho}^{(2)}_{AB^*}\) - k S_2\(\rho_{AB}\),
\end{align}
where we remind the reader that the above formulas are purely from the boundary perspective, and we have not used holography in deriving them.
With the holographic description of the state $\ket{\rho_{AB}}$ in hand, it is now straightforward to obtain the ERN by applying the modified cosmic brane proposal \Eqref{eq:MCS}.
Since $k\geq1$ and $k<1$ have qualitatively different behaviors, we will discuss them separately.

\subsubsection{$k\geq1$}

For the $k\geq1$ ERN we obtain
\begin{align}\label{eq:main1}
    \mathcal{E}^{(2k)}(A:B) = -k S_2(\rho_{AB})+
        \displaystyle \max_{A_1,A_2}\, \max_{i=1,2}\( k\log p^{(2)}\(A_1,A_2\) +\frac{(1-k)}{4G}A_i\)
\end{align}
where $p^{(2)}\(A_1,A_2\)$ is the probability distribution over the areas of the two candidate HRT surfaces, $\g_{A}\cup \g_{B^*}$ and $\g_{A^*}\cup \g_{B}$, for subregion $AB^*$ in the doubled state $\ket{\rho_{AB}}$.\footnote{We assume that other candidate HRT surfaces (which would lead to a connected phase for the entanglement wedge of $AB^*$) are not relevant. We will give an argument for this in the case of $k\leq 1$ at the end of this subsection.}
Note that \Eqref{eq:main1} applies to holographic states with arbitrary area distributions, but we will now focus on the case of greatest interest: holographic states prepared by a smooth gravitational path integral.

As discussed before, we can equivalently apply the original cosmic brane proposal for $k\geq1$ by swapping the order of maximization in \Eqref{eq:main1}.
By symmetry, there are two degenerate saddles (one for each of $i=1,2$) and we can consider either of them at leading order.
From the maximization conditions \Eqref{eq:max11} and \Eqref{eq:max21}, we see that a cosmic brane of tension $T_k$ is inserted at either $\g_1=\g_A\cup\g_{B^*}$ or $\g_2=\g_{A^*}\cup\g_B$.
Let us label the saddle that solves these maximization conditions as $\hat{\mathcal{B}}_2^{(k)}$ which satisfies $\pa \hat{\mathcal{B}}_2^{(k)}=M_2^{AB}$ and has conical defects of opening angle $\frac{2\pi}{k}$ at the surfaces $\g_A$ and $\g_{B^*}$ (or equivalently at $\g_{A^*}$ and $\g_{B}$).
Recalling \Eqref{eq:prob2} for the probability distribution over areas, the fact $\langle \rho_{AB}|\rho_{AB}\rangle=\Tr\(\rho_{AB}^2\)$, and the definition of the R\'enyi entropy, we arrive at
\begin{equation}\label{eq:km1}
    \mathcal{E}^{(2k)}(A:B)= k \(2I\[\mathcal{B}_1\] - I\[\hat{\mathcal{B}}_2^{(k)}\]\).
\end{equation}
Following Ref.~\cite{2021JHEP...06..024D}, we can define $I\[\hat{\mathcal{B}}_2^{(k)}\]\equiv I\(M_2^{AB},\g_A^{(2\pi/k)},\g_{B^*}^{(2\pi/k)}\)$ to emphasize the boundary conditions associated to the on-shell action and rewrite \Eqref{eq:km1} as
\begin{equation}\label{eq:km2}
    \mathcal{E}^{(2k)}(A:B)= k \[2I\(M_1\) - I\(M_2^{AB},\g_A^{(2\pi/k)},\g_{B^*}^{(2\pi/k)}\)\].
\end{equation}
This is precisely the result of Ref.~\cite{2021JHEP...06..024D} which made this proposal for general holographic states after obtaining their RTN results.
We have provided a boundary argument here for justifying their proposal for general states (in the sense of converting the calculation for the ERN to one for R\'enyi entropies on the boundary).
Moreover, we will see that their proposal only agrees with ours for $k\geq1$ since the original cosmic brane proposal fails for $k<1$.

For integer $k$, it is clear from \Eqref{eq:central_identity} and the standard assumption of replica symmetry in the R\'enyi entropy calculation that the saddle computing $\Tr \[\(\rho_{AB}^{T_B} \)^{2k}\]$ is guaranteed to preserve a $\mathbb{Z}_k$ symmetry.
Whether it preserves the full $\mathbb{Z}_{2k}$ symmetry depends on whether the location of the cosmic branes preserves the remaining $\mathbb{Z}_2$ symmetry.
We will show multiple examples in \secref{sec:PSSY} and \secref{sec:intervals} where the remaining $\mathbb{Z}_2$ symmetry is indeed broken, and in general we expect to always have a codimension-0 region in the parameter space where this occurs, although this is difficult to prove.

Before moving on to $k<1$, we would briefly like to mention a related quantity called the refined R\'enyi negativity (RRN)\footnote{As mentioned earlier, we will focus on the even case, so we will refrain from using the long name ``refined even R\'enyi negativity''.} given by
\begin{align}
    \tilde{\mathcal{E}}^{(2k)}(A:B) &= -k^2 \partial_k\left(\frac{1}{k} \mathcal{E}^{(2k)}(A:B) \right).\end{align}
The RRN is convenient to study because it will turn out to have a relatively simple geometric dual.
Using \Eqref{eq:central_identity}, the RRN simply becomes a refined R\'enyi entropy \cite{2016NatCo...712472D} 
\begin{align}
    \tilde{\mathcal{E}}^{(2k)}(A:B) &= \tilde{S}_k \left(\bar{\rho}^{(2)}_{AB^*}\right),\qqu
    \tilde{S}_k := k^2 \pa_k \(\fr{k-1}{k} S_k\),
\end{align}
which can also be computed using either the original or the modified cosmic brane proposal for $k\leq 1$.
A particularly useful limit for later purposes will be the RRN at $k=1$, where the refined R\'enyi entropy becomes an entanglement entropy and the bulk dual is given by
\begin{equation}
    \tilde{\mathcal{E}}^{(2)}(A:B) = \frac{\text{Area}\(\g_A\cup\g_{B^*}:\mathcal{B}_2\)}{4G},
\end{equation}
which is the area of probe extremal surfaces in the geometry corresponding to the doubled state.
By symmetry, of course, the same area also corresponds to $\g_{A^*}\cup\g_B$.

\subsubsection{$k<1$}

For the $k<1$ ERN we obtain
\begin{align}\label{eq:main2}
    \mathcal{E}^{(2k)}(A:B) = -k S_2(\rho_{AB})+
        \displaystyle \max_{A_1,A_2}\, \min_{i=1,2}\( k\log p^{(2)}\(A_1,A_2\) +\frac{(1-k)}{4G}A_i\).
\end{align}
In particular for the negativity $\(k=1/2\)$, we obtain
\begin{align}\label{eq:mainresult}
    \mathcal{E}(A:B) = -\frac{1}{2}S_2(\rho_{AB})+\frac{1}{2}\displaystyle \max_{A_1,A_2}\, \min_{i=1,2}\(\log p^{(2)}\(A_1,A_2\) +\frac{A_i}{4G}\).
\end{align}
As a sanity check, we can see that our results agree with those of Ref.~\cite{2021JHEP...06..024D} for fixed-area states as expected.
For simplicity, we will only look at the negativity.
For fixed-area states, the R\'enyi spectrum for $\rho_{AB}$ is flat \cite{2019JHEP...05..052A,2019JHEP...10..240D,2020JHEP...03..191D}.
Moreover, $p^{(2)}(A_1,A_2)$ is a probability distribution sharply localized at $A_1=A_2=4G(S(\rho_A)+S(\rho_B))$, since the areas of these surfaces are fixed; see, e.g., \figref{splitting_fig}.
Using this in \Eqref{eq:mainresult}, we obtain $\mathcal{E}(A:B)=\frac{1}{2}I(A:B)$ at leading order, in agreement with Ref.~\cite{2021JHEP...06..024D}.

We now discuss the failure of the original cosmic brane proposal.
In the special case of the calculation of ERN, this will in fact be quite generally true as we now show.
The special feature of the calculation of ERN is that the state $\ket{\rho_{AB}}$ has a $\mathbb{Z}_2$ symmetry.
Assuming the doubled state is in the connected phase, for any candidate RT surface that breaks the $\mathbb{Z}_2$ symmetry, we have another candidate RT surface arising from its $\mathbb{Z}_2$ image.\footnote{The results of Ref.~\cite{2021JHEP...06..024D} remain unchanged for the disconnected phase, so we will not focus on it.}
Moreover, because of the $\mathbb{Z}_2$ symmetry in $\ket{\rho_{AB}}$, we have the symmetry $p^{(2)}\(A_1,A_2\)=p^{(2)}\(A_2,A_1\)$.
We can now compare the original cosmic brane proposal and the modified cosmic brane proposal in this setting.

To do so, we will first establish some notation borrowed from Ref.~\cite{Dong:2023bfy}.
Let $f_i\equiv k \log p^{(2)}\(A_1,A_2\)+(1-k)\frac{A_i}{4G}$ for $i=1,2$.
Let $\vec{A^{(i)}}=\(A_1^{(i)},A_2^{(i)}\)$ be the point in the $(A_1, A_2)$ parameter space that maximizes $f_i$ subject to the minimality constraint $A_i\leq A_{3-i}$.
Further define $\vec{\tilde{A}^{(i)}}=\(\tilde{A}_1^{(i)},\tilde{A}_2^{(i)}\)$ to be the point in the parameter space that maximizes $f_i$ without any constraint.
Each of the above points in the parameter space depends on $k$.

For the ERN, we can then prove the following theorem:
\begin{nthm}\label{theorem}
The original cosmic brane proposal fails to correctly compute $\mathcal{E}^{(2k)}(A:B)$ for $k<1$ unless $\vec{\tilde{A}^{(1)}}=\vec{\tilde{A}^{(2)}}$ and they both lie on the diagonal $A_1=A_2$.
\end{nthm}
\begin{proof}
If neither $\vec{A^{(1)}}=\vec{\tilde{A}^{(1)}}$ nor $\vec{A^{(2)}}=\vec{\tilde{A}^{(2)}}$, then it is clear from Theorem~2 of Ref.~\cite{Dong:2023bfy} that the original cosmic brane proposal fails.
If $\vec{A^{(1)}}=\vec{\tilde{A}^{(1)}}$, then by the $\mathbb{Z}_2$ symmetry, we also have $\vec{A^{(2)}}=\vec{\tilde{A}^{(2)}}$.
Moreover, for $k<1$, we repeat the argument proving Lemma~1 of Ref.~\cite{Dong:2023bfy} and find
\begin{align}\label{eq:proof1}
        f_1\(\vec{\tilde{A}^{(1)}}\) = f_1\(\vec{A^{(1)}}\) &= k \log p^{(2)}\(A_1^{(1)},A_2^{(1)}\)+(1-k)\frac{A_1^{(1)}}{4G}\\\label{eq:proof12}
        &\leq k \log p^{(2)}\(A_1^{(1)},A_2^{(1)}\)+(1-k)\frac{A_2^{(1)}}{4G}=f_2\(\vec{A^{(1)}}\)\\
        &\leq f_2\(\vec{\tilde{A}^{(2)}}\),
\end{align}
where the second line uses $k<1$ and the fact that $\vec{A^{(1)}}$ by definition lies within the constrained domain $A_1\leq A_2$, and the third line uses the fact that $\vec{\tilde{A}^{(2)}}$ is the unconstrained maximum of $f_2$.
Due to the $\mathbb{Z}_2$ symmetry, the same argument can be repeated with the two candidate RT surfaces swapped to obtain $f_2\(\vec{\tilde{A}^{(2)}}\)\leq f_1\(\vec{\tilde{A}^{(1)}}\)$, thus implying equality.
The condition of equality implies that $\vec{\tilde{A}^{(1)}}$ and $\vec{\tilde{A}^{(2)}}$ both lie on the diagonal $A_1=A_2$.
Moreover, on the diagonal, the functions $f_i$ are identical and thus, the optimums must be the same, i.e., $\vec{\tilde{A}^{(1)}}=\vec{\tilde{A}^{(2)}}$.
\end{proof}

Theorem~\ref{theorem} shows that the original cosmic brane proposal fails generically, with the exception being the case where the two cosmic brane saddles are identical and they both have two exactly degenerate RT surfaces.
This exceptional case does happen, for instance, for fixed-area states, but for general holographic states, the modified cosmic brane proposal becomes crucial.

So far we only used the $\mathbb{Z}_2$ symmetry of the doubled state, but we will now use its explicit form in order to obtain a stronger result.
We will show that for $k<1$, and in particular for the negativity, the optimum in \Eqref{eq:main2} and \Eqref{eq:mainresult} is always achieved on the diagonal $A_1=A_2$.
This phase was called the diagonal phase in Ref.~\cite{Dong:2023bfy} and our theorem below amounts to proving that the saddle computing the negativity is always in the diagonal phase.
This then allows us to write down a simpler holographic dual for negativity.

We start by proving the following lemma:
\begin{nlemma}\la{lemmatrproj}
Let $\r$ be a density operator and $P_1$, $P_2$ be two projection operators. Then
\be\la{trproj}
\Tr(\r P_1 \r P_2) \leq \sqrt{ \Tr(\r P_1 \r P_1) \Tr(\r P_2 \r P_2) }.
\ee
\end{nlemma}

\begin{proof}

Using the fact that the left-hand side of \Eqref{trproj} is the trace of a positive operator $\(P_1\r P_2\)^\dagger\(P_1 \r P_2\)$ and thus is a non-negative real number, we have
\begin{align}
    \Tr(\r P_1 \r P_2) &= \Tr\(\(\sqrt{\r} P_1\sqrt{\r}\) \(\sqrt{\r} P_2\sqrt{\r}\)\)\\
    &\leq \sqrt{\Tr(\r P_1 \r P_1)\Tr(\r P_2 \r P_2)},
\end{align}
where in the second line, we have used the Cauchy-Schwarz inequality for the inner product $\bra{\sigma_1}\ket{\sigma_2}=\Tr[\sigma_1^\dagger\sigma_2]$ in the doubled Hilbert space discussed earlier.\footnote{We thank Francesco Mele for suggesting this simpler proof compared to the previous version.}

\end{proof}

This leads us to our next main result:

\begin{nthm}\la{diagthm}
For $0<k<1$ and at leading order in $G$,
\be\la{maximin}
\max_{A_1,A_2}\, \min_{i=1,2}\( k\log p^{(2)}\(A_1,A_2\) +\frac{(1-k)}{4G}A_i\)
\ee
is achieved on the diagonal $A_1=A_2$.\footnote{Note that in cases of degenerate optima, not all optima need to be on the diagonal, but the theorem guarantees that at least one optimum is on the diagonal.}
\end{nthm}

\begin{proof}
Using \Eqref{eq:prob}, we find
\be
p^{(2)}\(A_1,A_2\) = N \lt\<\r_{AB}\lt|P_{|\g_{A}\cup \g_{B^*}|=A_1,\, |\g_{A^*}\cup \g_{B}|=A_2}\rt|\r_{AB}\rt\>,
\ee
where $N=\braket{\r_{AB}}^{-1}$ is a normalization constant and $|\g_{A}\cup \g_{B^*}|=A_1$ means fixing the area of $\g_{A}\cup \g_{B^*}$ to $A_1$.
On the doubled Hilbert space, the inner product is given by  $\langle C|D\rangle=\Tr\(C^\dagger D\)$, and the action of operators on $AB$ ($A^*B^*)$ is given by left (right) multiplication on $\r_{AB}$\cite{2019arXiv190500577D}.
Using this, we can rewrite the expectation value as a trace:
\begin{align}
    &p^{(2)}\(A_1,A_2\) \\
=& N \int da_1 da_2 db_1 db_2 \d(a_1+b_2-A_1) \d(a_2+b_1-A_2) \lt\<\r_{AB}\lt|P_{|\g_{A}|=a_1,\, |\g_{B}|=b_1}\, P_{|\g_{A^*}|=a_2,\, |\g_{B^*}|=b_2}\rt|\r_{AB}\rt\>\\
=& N \int da_1 da_2 db_1 db_2 \d(a_1+b_2-A_1) \d(a_2+b_1-A_2) \Tr \( \r_{AB} \, P_{|\g_{A}|=a_1,\, |\g_{B}|=b_1} \, \r_{AB} \, P_{|\g_{A}|=a_2,\, |\g_{B}|=b_2}\).
\la{p2tr}
\end{align}

At leading order in $G$, the integral is well approximated by the maximal value
\be
p^{(2)}\(A_1,A_2\) \ap N \max_{\substack{a_1,a_2,b_1,b_2\\ a_1+b_2=A_1,\, a_2+b_1=A_2}} \Tr \( \r_{AB} \, P_{|\g_{A}|=a_1,\, |\g_{B}|=b_1} \, \r_{AB} \, P_{|\g_{A}|=a_2,\, |\g_{B}|=b_2}\).
\ee
Let $(\bar a_1,\bar a_2,\bar b_1,\bar b_2)$ be a location where this maximum is achieved; they depend on $A_1$, $A_2$ and satisfy the constraints
\be\la{constraints}
\bar a_1+\bar b_2=A_1,\qu
\bar a_2+\bar b_1=A_2.
\ee
Using Lemma \ref{lemmatrproj}, we find
\ba
p^{(2)}\(A_1,A_2\) &\ap N \Tr \( \r_{AB} \, P_{|\g_{A}|=\bar a_1,\, |\g_{B}|=\bar b_1} \, \r_{AB} \, P_{|\g_{A}|=\bar a_2,\, |\g_{B}|=\bar b_2}\)\\
&\leq N \sqrt{ \prod_{i=1,2} \Tr \( \r_{AB} \, P_{|\g_{A}|=\bar a_i,\, |\g_{B}|=\bar b_i} \, \r_{AB} \, P_{|\g_{A}|=\bar a_i,\, |\g_{B}|=\bar b_i}\) }\\
&\leq \sqrt{ \prod_{i=1}^2 p^{(2)}\(\bar a_i+\bar b_i,\bar a_i+\bar b_i\) },
\la{p2bound}
\ea
where in the last step we have used
\be
p^{(2)}\(\bar a_i+\bar b_i,\bar a_i+\bar b_i\) \geq N \Tr \( \r_{AB} \, P_{|\g_{A}|=\bar a_i,\, |\g_{B}|=\bar b_i} \, \r_{AB} \, P_{|\g_{A}|=\bar a_i,\, |\g_{B}|=\bar b_i}\),
\ee
since the right-hand side is one contribution to the left-hand side according to \er{p2tr}, and all contributions are nonnegative because $\Tr\(\r P_1 \r P_2\) = \Tr\[(P_1\r P_2)^\dag (P_1\r P_2)\] \geq 0$.

Note that \er{p2bound} holds for an arbitrary $(A_1, A_2)$. We can now show that $(A_1, A_2)$ cannot give a more optimal value to \er{maximin} than a corresponding point on the diagonal. Due to the symmetry $p^{(2)}\(A_1,A_2\)=p^{(2)}\(A_2,A_1\)$, the image location $(A_2, A_1)$ is degenerate with $(A_1, A_2)$. Thus, without loss of generality, we assume $A_1\leq A_2$.
Since $A_1\leq A_2$, the contribution of $(A_1, A_2)$ is $f_1\(A_1,A_2\)$, where we remind the reader that we defined $f_i\equiv k \log p^{(2)}\(A_1,A_2\)+(1-k)\frac{A_i}{4G}$ for $i=1,2$.
Rewriting \er{p2bound} as
\be
\log p^{(2)}\(A_1,A_2\) \leq \fr{1}{2} \sum_{i=1}^2 \log p^{(2)}\(\bar a_i+\bar b_i,\bar a_i+\bar b_i\)
\ee
and using \er{constraints} to find
\be
\sum_{i=1}^2 (\bar a_i+\bar b_i) =A_1+A_2 \geq 2A_1,
\ee
we arrive at
\begin{align}\la{f1bound}
f_1\(A_1,A_2\) &= k\log p^{(2)}\(A_1,A_2\) +\frac{(1-k)}{4G}A_1\\
&\leq \fr{1}{2} \sum_{i=1}^2 \[k\log p^{(2)}\(\bar a_i+\bar b_i,\bar a_i+\bar b_i\) +\frac{(1-k)}{4G}\(\bar a_i+\bar b_i\)\]\\
&=\fr{1}{2} \sum_{i=1}^2 f_1\(\bar a_i+\bar b_i,\bar a_i+\bar b_i\)\\
&\leq \max_{i=1,2} f_1\(\bar a_i+\bar b_i,\bar a_i+\bar b_i\).\la{f1boundlast}
\end{align}
Thus we have found a point on the diagonal which contributes no less than $(A_1, A_2)$. Since we showed this for arbitrary $(A_1, A_2)$, we have shown that \er{maximin} is achieved on the diagonal.
\end{proof}

This is a very powerful result since it highly simplifies the calculation of ERN for $k<1$ and in particular, the negativity.
Anticipating that the optimum for the modified cosmic brane proposal is achieved on the diagonal, we can remove the inner minimization in \Eqref{eq:main2} and replace $A_i$ with $\frac{A_1+A_2}{2}$, finding
\be\la{eq:mainsimple}
    \mathcal{E}^{(2k)}(A:B) = -k S_2(\rho_{AB})+
        \displaystyle \max_{A_1,A_2}\( k\log p^{(2)}\(A_1,A_2\) +\frac{(1-k)}{8G}(A_1+A_2)\).
\ee
To see that this maximization is also achieved on the diagonal and thus gives the same result as \Eqref{eq:main2}, we use an argument similar to (\ref{f1bound})--(\ref{f1boundlast}):
\begin{align}
\frac{f_1\(A_1,A_2\)+f_2\(A_1,A_2\)}{2} &= k\log p^{(2)}\(A_1,A_2\) +\frac{(1-k)}{8G}(A_1+A_2)\\
&\leq \fr{1}{2} \sum_{i=1}^2 \[k\log p^{(2)}\(\bar a_i+\bar b_i,\bar a_i+\bar b_i\) +\frac{(1-k)}{4G}\(\bar a_i+\bar b_i\)\]\\
&\leq \max_{i=1,2} f_1\(\bar a_i+\bar b_i,\bar a_i+\bar b_i\).
\end{align}

The maximization conditions in \er{eq:mainsimple} then become \cite{Dong:2023bfy}
\begin{align}\label{eq:max3}
     \frac{\pa I\[g_{A_1,A_2}\]}{\pa A_1}&= \frac{1-k}{8 k G},\\
    \frac{\pa I\[g_{A_1,A_2}\]}{\pa A_2}&= \frac{1-k}{8 k G}\label{eq:max4},
\end{align}
where we remind the reader that $I\[g_{A_1,A_2}\]$ is the action of the fixed-area saddle.

Analogous to \Eqref{eq:max11} and \Eqref{eq:max21}, \Eqref{eq:max3} and \Eqref{eq:max4} are precisely the bulk equations of motion obtained by inserting cosmic branes of tension $\frac{T_k}{2}=\frac{k-1}{8kG}$ at the surfaces $\g_A \cup \g_{B^*}$ and $\g_{A^*}\cup\g_B$.
In other words, the proposal is to insert cosmic branes of half the usual tension considered in the original cosmic brane proposal at each of the surfaces $\g_A,\g_B,\g_{A^*}$ and $\g_{B^*}$.
For a holographic state with a sufficiently smooth area distribution, the RHS of \Eqref{eq:max3} and \Eqref{eq:max4} can be related to the deficit angle at the surfaces \cite{Dong:2023bfy} leading to a saddle with opening angle $\pi+\frac{\pi}{k}$ at each of the candidate RT surfaces which are degenerate by symmetry.

We again label the resulting saddle $\hat{\mathcal{B}}_2^{(k)}$, with the understanding that it satisfies $\pa \hat{\mathcal{B}}_2^{(k)} = M_2^{AB}$ and has conical defects of opening angle $\pi+\frac{\pi}{k}$ at each of the surfaces $\g_A$, $\g_B$, $\g_{A^*}$ and $\g_{B^*}$.
Then, similar to the case of $k\geq 1$, \Eqref{eq:main2} simplifies to
\begin{equation}\label{eq:kl1}
    \mathcal{E}^{(2k)}(A:B)= k \(2 I\[\mathcal{B}_1\] - I\[\hat{\mathcal{B}}_2^{(k)}\]\).
\end{equation}
Curiously, the diagonal phase restores the $\mathbb{Z}_2$ symmetry that was lost for $k\geq1$, allowing us to further quotient by the $\mathbb{Z}_2$ symmetry.
Using this to rewrite \Eqref{eq:kl1} while also emphasizing its boundary conditions, we obtain
\begin{equation}
    \mathcal{E}^{(2k)}(A:B)= 2 k \[I\(M_1\) - I\(M_1,\g_{AB}^{(\pi)},\(\g_{A}\cup\g_B\)^{(\pi+\pi/k)}\)\].
\end{equation}

In particular, we have arrived at a remarkably simple geometric prescription for the negativity summarized by
\begin{align}\label{eq:mainresult2}
    \mathcal{E}(A:B) = I\(M_1\) - I\(M_1,\g_{AB}^{(\pi)},\(\g_{A}\cup\g_B\)^{(3\pi)}\),
\end{align}
where the second term corresponds to a gravitational saddle with boundary conditions set by the original state and has conical defects of opening angle $\pi$ at $\g_{AB}$ and $3\pi$ at $\g_A$ and $\g_B$.

Throughout this subsection, we have restricted our attention to the two candidate HRT surfaces, $\g_{A}\cup \g_{B^*}$ and $\g_{A^*}\cup \g_{B}$, for subregion $AB^*$ in the doubled state $\ket{\rho_{AB}}$. In principle, when applying the modified cosmic brane proposal we should also include other candidate HRT surfaces (which would lead to a connected phase for the entanglement wedge of $AB^*$). They include, for example in the case of Fig.~\re{fig:torus}, the union of a line connecting the left endpoint of $A$ to the right endpoint of $B^*$ and a line connecting the right endpoint of $A$ to the left endpoint of $B^*$ (as well as analogues with higher winding numbers). Near $k=1$ (the case without backreaction), it is easy to see generally that these surfaces are subdominant because they have larger areas than $\g_{A}\cup \g_{B^*}$ or $\g_{A^*}\cup \g_{B}$ due to the $\Z_2$ symmetry \cite{2019arXiv190500577D}. For general $k<1$, an analogue of Theorem~\re{diagthm} shows that the ERN is dominated by the diagonal for these surfaces as well (i.e., they have the same area as their $\Z_2$ images). Then using this $\Z_2$ symmetry and an argument similar to the holographic proof of strong subadditivity \cite{Headrick:2007km}, we find that these surfaces are subdominant to the two surfaces studied above and can be ignored.

\section{PSSY Model}
\label{sec:PSSY}
\begin{figure}
    \centering
    \includegraphics[scale=0.6]{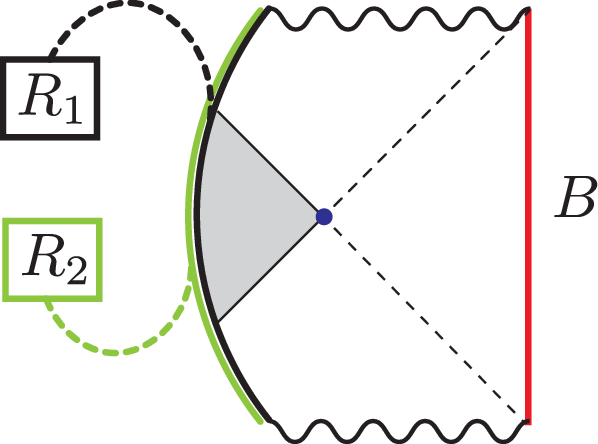}
    \caption{The PSSY model consists of a JT gravity black hole coupled to ETW branes with flavor indices (black and green) entangled with auxiliary radiation systems $R_1$ and $R_2$.}
    \label{fig:WC}
\end{figure}

In this section, we analyze the negativity for the PSSY model of black hole evaporation \cite{2019arXiv191111977P}.
This problem was previously studied in Ref.~\cite{2021arXiv211011947D}, and related models have been studied using the equilibrium approximation in Refs.~\cite{2021arXiv211200020V,2022PhRvL.129f1602V}.
We will find perfect agreement with these previous results.

The PSSY model is a theory of Jackiw-Teitelboim (JT) gravity coupled to end-of-the-world (ETW) branes.
The ETW branes are entangled with two auxiliary radiation systems $R_1$ and $R_2$.
The state of the whole system as depicted in \figref{fig:WC} is
\begin{equation}
    \ket{\Psi} = \frac{1}{\sqrt{k}} \sum_{i=1}^{k_1}\sum_{j=1}^{k_2} \ket{i}_{R_1}\ket{j}_{R_2}\ket{\psi_{ij}(\beta)}_{B},
\end{equation}
where $\ket{\psi_{ij}(\beta)}_{B}$ is the state of the black hole system $B$ at inverse temperature $\beta$ with the ETW brane chosen to be of sub-flavors $i$ and $j$ respectively.

To apply our proposal, we need to consider the doubled state $\ket{\rho_{R_1R_2}}$ of the radiation systems in the PSSY model.
This problem was studied in Ref.~\cite{2022JHEP...06..089A}, which found that the gravitational description of the doubled state is as shown in \figref{fig:R1R2}.
There are two phases depending on whether $k=k_1 k_2$ is smaller/larger compared to the parameter $S_0$ in JT gravity.
More precisely, the transition is determined by the dominant saddle for the second R\'enyi entropy of the thermal black hole, or equivalently the black hole entropy at inverse temperature $2\beta$.
When the radiation is in the disconnected phase for small $k$, the doubled state simply involves a doubled copy of the radiation system.
On the other hand, for large $k$ the radiation has an island and is in the connected phase, in which case the doubled state includes a closed universe.
The closed universe involves two copies of the island obtained in the computation of the second R\'enyi entropy.

\begin{figure}
    \centering
    \includegraphics[scale=0.6]{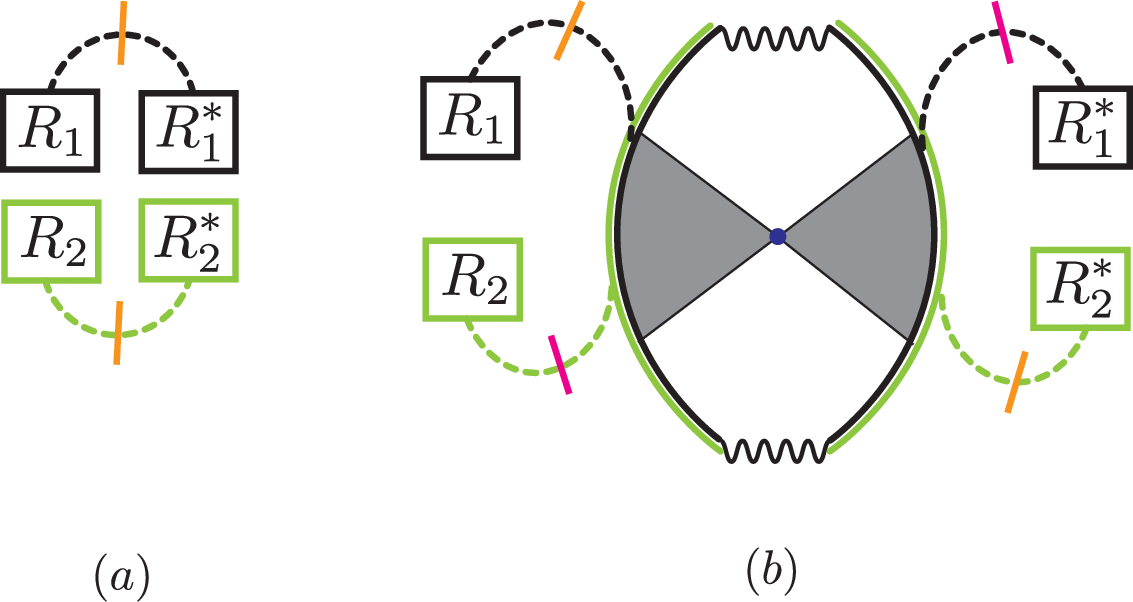}
    \caption{The doubled state $\ket{\rho_{R_1 R_2}}$ in (a) the disconnected phase and (b) the connected phase (where we have further assumed a so-called $\tau$ phase where neither $k_1$ nor $k_2$ is too large). The RT surfaces computing the RRN are depicted in orange in both geometries. In the connected phase, the magenta surface is degenerate with the orange surface.}
    \label{fig:R1R2}
\end{figure}

We can now use the gravitational description of the doubled state to compute the ERN and RRN.
For simplicity, we will only focus on the RRN since it has an easier holographic dual -- the areas of RT surfaces anchored to $R_1 R_2^*$ as depicted in \figref{fig:R1R2}.
Furthermore, we will approximate the area of the black hole by $S_0$ and ignore $\frac{1}{\beta}$ corrections for simplicity, although our results agree with those of Ref.~\cite{2021arXiv211011947D} even after including them.
In this approximation, all the entanglement spectra are flat.

In the disconnected phase, we find $\tilde{\mathcal{E}}^{(2k)}(R_1:R_2)=\log k$.
This is the identity phase of Ref.~\cite{2021arXiv211011947D}.
In the connected phase, $\tilde{\mathcal{E}}^{(2k)}(R_1:R_2)$ depends on how large $k_1,k_2$ are relative to each other.
If they are comparable, then we have $\tilde{\mathcal{E}}^{(2k)}(R_1:R_2)=\log k$. This is the so-called $\tau$ phase of Ref.~\cite{2021arXiv211011947D}.
If instead $\log k_1\geq \log k_2+S_0$ (which is the cyclic phase of Ref.~\cite{2021arXiv211011947D}), then we obtain $\tilde{\mathcal{E}}^{(2k)}(R_1:R_2)=2\log k_2 +S_0$.
Similarly, we obtain the anti-cyclic phase of Ref.~\cite{2021arXiv211011947D} by swapping $1\leftrightarrow2$ in the previous phase.

In the previous example, we see that the modified cosmic brane proposal did not play a key role.
This would have been true even if we included $\frac{1}{\beta}$ corrections.
This is because the only situations with two degenerate candidate surfaces, i.e., the ones which break the full replica symmetry, involve a flat entanglement spectrum as shown in \figref{fig:R1R2}.
\begin{figure}
    \centering
    \includegraphics[scale=0.6]{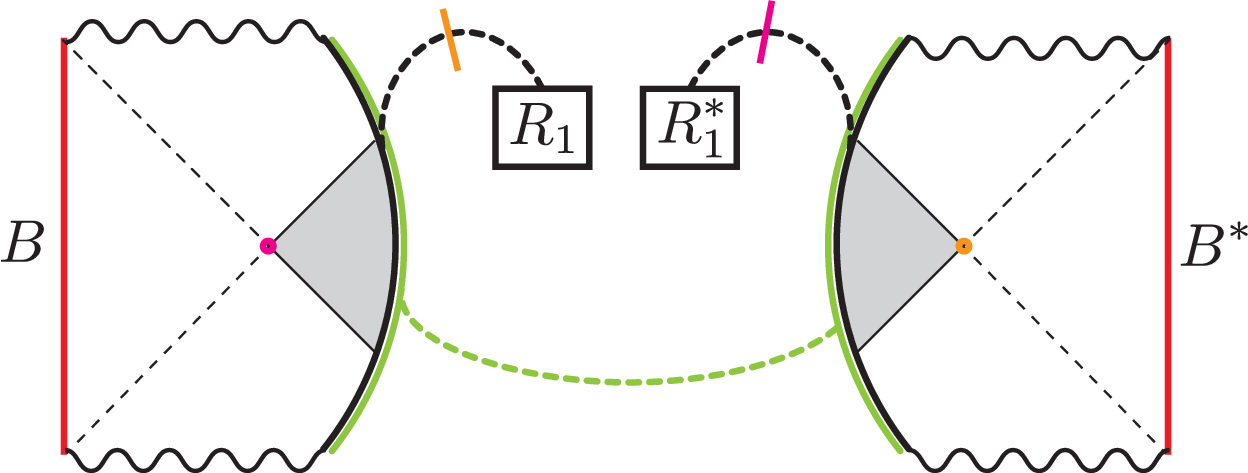}
    \caption{The doubled state $\ket{\rho_{BR_1}}$ in the connected phase is depicted. The degenerate RT surfaces computing the RRN are shown in orange and magenta.}
    \label{fig:BR1}
\end{figure}

We can instead consider the negativity between $R_1$ and $B$ as the simplest setting where the modified cosmic brane proposal becomes important.
In this case, the doubled state in the connected phase is depicted in \figref{fig:BR1}.
In the phase where neither $k_1$ nor $S_0$ is too large, we find that the RT surface is as shown in \figref{fig:BR1}.
In this case, there are non-trivial fluctuations in the area spectrum due to the thermal fluctuations of the black hole area.
The modified cosmic brane proposal thus becomes important in this situation.
This setup was not analyzed previously in the literature, and it would be interesting in the future to compute it using a full resolvent calculation to check the validity of our modified cosmic brane proposal.

\section{Two Intervals in the Vacuum State}
\label{sec:intervals}

In Refs.~\cite{2019PhRvD..99j6014K, 2019PhRvL.123m1603K}, 2D holographic CFT calculations were presented that provided evidence for the conjecture that the negativity was related to the area of a cosmic brane located at the entanglement wedge cross section. Given that this disagrees with our results, it is helpful to revisit these calculations to identify the assumptions that do not hold.

For simplicity, we only consider the case of the negativity between two disjoint intervals, $[0,x]$ and $[1,\infty]$, in the vacuum state. The even moments of the partially transposed density matrix are given by the following four-point function in the product of $2k$ copies of the original CFT \cite{2012PhRvL.109m0502C, 2013JSMTE..02..008C}
\begin{align}
\label{LN_four_point}
    \Tr\left( \rho_{AB}^{T_B}\right)^{2k} = \left \langle \sigma_{2k}(0) \sigma_{2k}^{-1}(x) \sigma_{2k}^{-1}(1) \sigma_{2k}(\infty)\right \rangle,
\end{align}
where $\sigma_{2k}$ and $\sigma_{2k}^{-1}$ are cyclic and anti-cyclic twist operators respectively which are scalar Virasoro primary operators with identical conformal dimensions
\begin{align}
    \Delta_{2k} = \frac{c}{12}\left( 2 k-\frac{1}{2k}\right).
\end{align}
To find the negativity, one analytically continues this correlation function to $k = 1/2$. 

\begin{figure}
    \centering
    \includegraphics[scale=0.45]{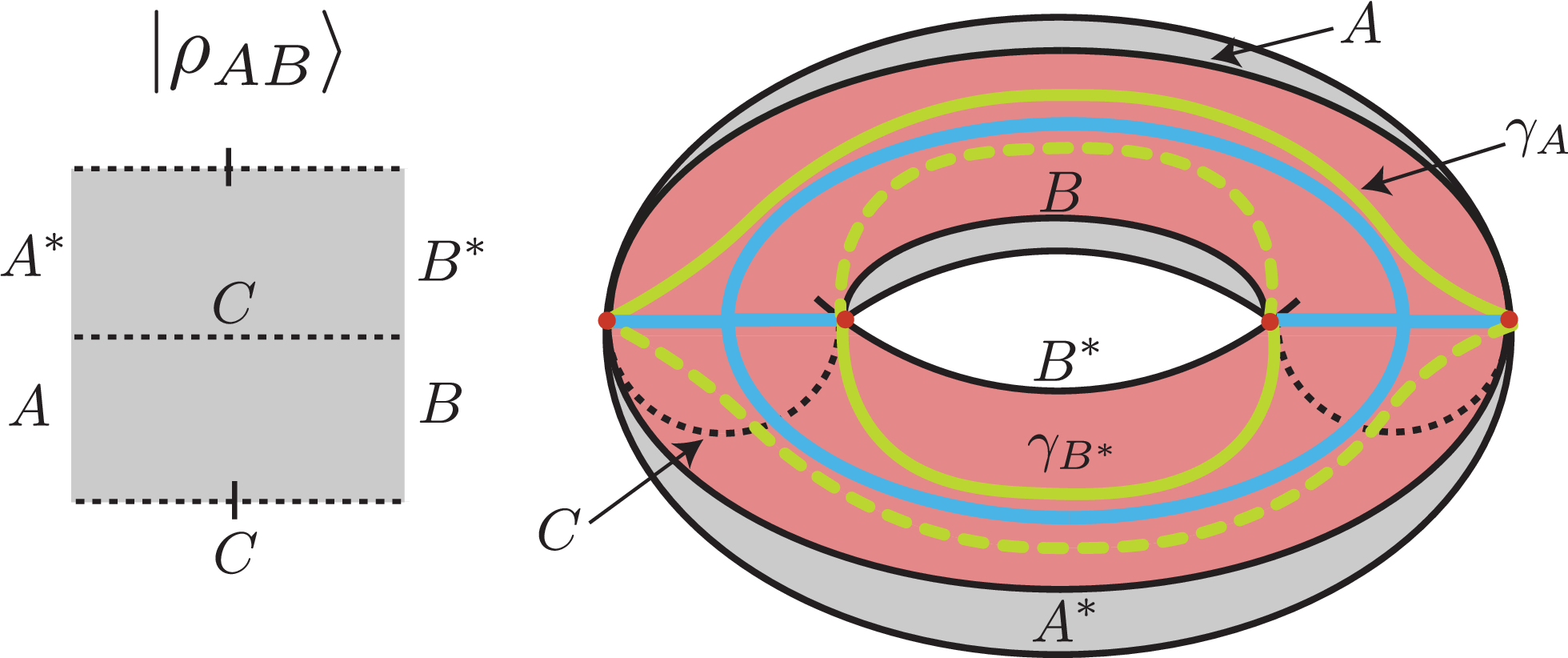}
    \caption{(left): The state $\rho_{AB}$ for $A,B$ chosen to be disjoint intervals in the vacuum state is computed by a square with open boundary conditions at the subregions $A,B,A^*,B^*$. (right): The bulk dual of the state $\ket{\rho_{AB}}_{AA^*BB^*}$ is the BTZ black hole geometry up to conformal transformations. The R\'enyi entropy involves computing $\Tr[(\rho^{(2)}_{AB^*})^k]$ whose bulk dual involves conical defects (green) sourced by twist operators (red). For $k\geq 1$, the cosmic brane is placed at either the solid or the dashed green lines(thus breaking the $\mathbb{Z}_{2k}$ replica symmetry at integer $k$), whereas for $k<1$ it is distributed over both surfaces. The $\mathbb{Z}_{2k}$ replica symmetric configuration with intersecting cosmic branes is shown in blue.}
    \label{fig:torus}
\end{figure}
Using our formalism, we are interested in the doubled state $\ket{\rho_{AB}}$, which can be computed from the path integral computing $\Tr\(\rho_{AB}^2\)$.
For the case of two intervals in the vacuum state, it is well known that this path integral is related to the torus partition function via a conformal transformation \cite{Headrick:2010zt}.
The leading bulk saddle for torus boundary conditions is given either by thermal AdS in the disconnected phase or the BTZ black hole in the connected phase.
In the disconnected phase, the negativity vanishes and the resulting saddle is fully replica symmetric.
In the connected phase, one finds replica symmetry breaking (RSB) as expected in general.
We depict the geometry of the doubled state along with the relevant RT surfaces computing the RRN in \figref{fig:torus}.

The configuration considered in the proposal of Ref.~\cite{2019PhRvL.123m1603K}, which has intersecting cosmic branes, is also shown in \figref{fig:torus} for comparison.
Ref.~\cite{2019PhRvL.123m1603K} considered this family of intersecting brane configurations with a $k$-dependent tension at $k\geq1$ and analytically continued the tension to $k<1$ in order to compute the negativity.
However, at integer $k$, it is crucial that any candidate saddle comes from a smooth parent space, with cosmic branes resulting from performing the quotient by $\mathbb{Z}_{k}$.
By understanding the possible brane intersections that can arise from a quotient of a smooth manifold, we are able to prove in Appendix~\ref{app:brane} that it is impossible to have an intersection of the type proposed by Ref.~\cite{2019PhRvL.123m1603K}.
Thus, this family of configurations should not be considered saddles for the negativity problem at any $k$.

The reason this configuration fails to be a saddle is simplest to understand when we consider the intervals to be adjacent, using a version of the argument presented in Ref.~\cite{Penington:2022dhr} for a different entanglement quantity.
In \figref{fig:triway}, we show the $\mathbb{Z}_{2k}$ quotient\footnote{We caution the reader that this is unlike most of our discussion in the rest of the paper where we consider $\mathbb{Z}_{k}$ quotients.}  of the putative fully ($\mathbb{Z}_{2k}$) replica symmetric saddle where three branes meet at a vertex.
The parent space can be obtained by gluing together $2k$ such copies in a manner specified by the permutations labelling different bulk regions in \figref{fig:triway}.
Using a radial coordinate $r$ along the branes (which has the range $[0,1]$), we find that the topology of this parent space is $\(\Sigma_{2k}\times \[0,1\]\)/\sim$, where $\Sigma_{2k}$ is the topology of the Riemann surface that computes the boundary replica partition function and $\sim$ identifies all the points at $r=0$.
Using the Riemann-Hurwitz formula, one finds that the genus of $\Sigma_{2k}$ is $k-1$.
Since the neighborhood of every point in a smooth manifold is topologically a ball, the parent space can be a smooth manifold only if $\Sigma_{2k}$ is a sphere.
Thus, it is clear that for $k\geq 2$, the replica symmetric configuration shown in \figref{fig:triway} cannot be a saddle.
Since this argument is local at the intersection, the same is true for multiple intervals. A more rigorous version of this argument is presented in Appendix~\ref{app:brane}.

\begin{figure}
    \centering
    \includegraphics[scale=0.45]{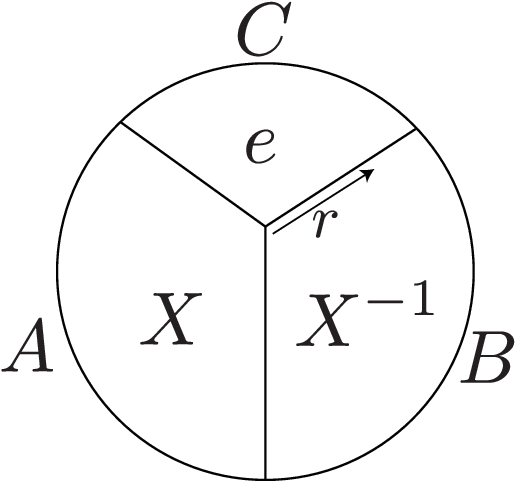}
    \caption{The quotient space description of the naive $\mathbb{Z}_{2k}$ replica symmetric bulk configuration with three intersecting branes (solid lines) is shown for $A,B$ being adjacent intervals in the vacuum state. $2k$ such copies are glued together in a manner specified by the permutations $\{X,X^{-1},e\}$ on the different regions to obtain the parent space; see \figref{splitting_fig} for details. The radial coordinate $r$ goes from the intersection point $r=0$ to the asymptotic boundary at $r=1$.}
    \label{fig:triway}
\end{figure}

Having presented our proposal, it is useful to understand why the CFT calculation in Ref.~\cite{2019PhRvL.123m1603K} failed.
When presented with a four-point function such as \eqref{LN_four_point}, it is usually convenient to perform a conformal block decomposition. When the intervals are close, we may take the $\sigma_{2k} \times \sigma_{2k}$ and $\sigma_{2k}^{-1}\times \sigma_{2k}^{-1}$ OPEs, expanding in the t-channel as
\begin{align}
\label{conf_block}
    \Tr\left( \rho_{AB}^{T_B}\right)^{2k} = \sum_p |C_{\sigma_{2k} \sigma_{2k} p}|^2 \mathcal{F}_p(1-x)\bar{\mathcal{F}}_p(1-\bar{x}),
\end{align}
where the sum is over primary fields, $\mathcal{F}$ is the Virasoro conformal block.
The vacuum state does not contribute to this sum because it has twist number $0$. Instead, the  primary field contributing to the sum with the lowest dimension is the ``double-twist'' operator that performs two consecutive cyclic permutations. 

In the large central charge limit, the conformal blocks approximately exponentiate as \cite{zamolodchikov1984conformal,2020JHEP...01..109B}
\begin{align}
    \mathcal{F}(x) \simeq \exp\left[-\frac{c}{6}f\(\frac{h_p}{c}, \frac{h_{2k}}{c}, x\) \right].
\end{align}
Under seemingly mild assumptions on OPE coefficients and the spectrum, one can argue \cite{2013arXiv1303.6955H} that the primary field in the OPE with the lowest conformal dimension then dominates the sum \Eqref{conf_block}. Namely, one assumes that the Cardy density of states times the OPE coefficients $|C_{\sigma_{2k}\sigma_{2k} p}|^2$ does not grow exponentially with $c$ faster than the suppression from the conformal block for a finite range of $z$. It is therefore this assumption that must break down. This suggests that the holographic formula for RRN may lead to interesting constraints on the $C_{\sigma_{2k}\sigma_{2k} p}$ OPE coefficients. A similar result is shown in Appendix~\ref{sub:PRMI} for the Petz R\'enyi mutual information, where an analogous assumption leads to an obviously wrong answer.

Nevertheless, sticking with this assumption, one may compute the conformal block with the double-twist operator as the intermediate state, which has conformal dimension $\Delta^{(2)}_{2k} = \frac{c}{6}\left(k-\frac{1}{k} \right)$, although this computation is generally difficult to do explicitly and analytically because the operator is heavy, i.e., $O(c)$. However, it may be evaluated numerically with arbitrarily high precision using Zamolodchikov's recursion relations \cite{zamolodchikov1987conformal,zamolodchikov1984conformal}.

The more general idea of Ref.~\cite{2019PhRvL.123m1603K} was to relate the negativity calculation to a calculation of the $(m,n)$-R\'enyi reflected entropy and then argue that the negativity is given by the entanglement wedge cross section, due to the known connection between reflected entropy and entanglement wedge cross section \cite{2019arXiv190500577D}. In particular, the $(m,n)$-R\'enyi reflected entropy is evaluated via a four-point function of twist operators
\begin{align}
    S_R^{(m,n)} = \frac{1}{1-n}\log \frac{Z_{m,n}}{Z_{m,1}^n}, \quad Z_{m,n} = \langle{\sigma_{g_A}(0)\sigma_{g_A^{-1}}(x)\sigma_{g_B}(1)\sigma_{g_B^{-1}}(\infty)}\rangle
\end{align}
in $mn$ copies of the original CFT. Labeling the copies from $1$ to $mn$, the permutations are given in cycle notion as
\begin{align}
    {g_A} &= \prod_{i = 1}^n \big(i, i+n,\dots,i+n(m/2-1),i+1+nm/2,
    \dots,i+1+n(m-1)\big),
    \\
    {g_B} &= \prod_{i = 1}^n \left(i, i+n,\dots,i+n(m-1)\right).
\end{align}
The relevant conformal dimensions are
\begin{align}
    \Delta_{g_A}= \Delta_{g_B}=\frac{cn}{12m}\(m^2-1\),\quad
     \Delta_{g_B g_A^{-1}}=\frac{c}{6n}\(n^2-1\).
\end{align}
For the following choice of $\(n,m\)$, the dimensions of the operators computing the four point function for negativity can be matched to the four point function for reflected entropy:
\begin{align}\label{eq:match}
    n=k,\qquad m=\frac{4k^2+\sqrt{32 k^4-8 k^2+1}-1}{4 k^2}.
\end{align}
The idea was then to use the fact that the dominant conformal block for the reflected entropy calculation is known to be related to a backreacted version of the entanglement wedge cross section by a path integral argument, and moreover the conformal blocks only depend on the relevant operator dimensions. Thus, the answers can be matched by using the identification \Eqref{eq:match}.

We would now like to demonstrate that our proposed saddle is strictly better than this replica symmetric configuration near $k=1$, which we have already argued quite generally from the gravitational side earlier. The calculation of the reflected entropy can be performed gravitationally by computing the replica partition function $Z_{m,n}$. The RRN is given by
\begin{equation}\label{eq:RRN2}
    \tilde{\mathcal{E}}^{(2k)} = -k^2 \pa_{k}\(\frac{1}{k} \log Z_{2k}\),
\end{equation}
where $Z_{2k}$ is the replica partition function for the negativity problem. We reproduce the answer from the replica symmetric configuration by assuming $Z_{2k}=Z_{m(k),n(k)}$, i.e., the replica partition functions for negativity and reflected entropy agree upon identification \Eqref{eq:match}. We will call the resulting RRN $\tilde{\mathcal{E}}^{(2k)}_{CFT}$, and \Eqref{eq:RRN2} becomes
\be
\tilde{\mathcal{E}}^{(2k)}_{CFT} = -k^2 \pa_{k}\(\frac{1}{k} \log Z_{m(k),n(k)}\).
\ee
Note that since $Z_2$ computes $\Tr\(\rho_{AB}^2\)$ where the two proposals agree, our proposed saddle dominating over the replica symmetric configuration would mean that $\tilde{\mathcal{E}}^{(2k)}_{CFT}$ is larger than $\tilde{\mathcal{E}}^{(2k)}$ computed using our proposal, which amounts to computing the areas of minimal surfaces anchored to $AB^*$ in the doubled state.

In order to test this, we need to compute $\tilde{\mathcal{E}}^{(2k)}_{CFT}$ gravitationally.
We can do this by using the Lewkowycz-Maldacena method \cite{2013JHEP...08..090L}. However, it is important to note that the Lewkowycz-Maldacena method implies
\begin{equation}
    -n^2\pa_n\(\frac{1}{n}\log Z_n\) = \frac{A_n}{4G}
\end{equation}
where $Z_n$ is the replica partition function for R\'enyi entropies and $A_n$ is the area of a cosmic brane with tension $T_n=\frac{n-1}{4nG}$. So for the $k=1$ RRN (corresponding to $m=2$, $n=1$), we have
\begin{align}
    \tilde{\mathcal{E}}^{(2)}_{CFT} &= \lt. -k^2 \pa_{k}\(\frac{1}{m n} \frac{m n}{k}\log Z_{m,n}\) \rt|_{k=1}\\
    &= \lt. -k^2 \(\pa_{k}\(\frac{m n}{k}\)\frac{\log Z_{m,n}}{mn}+\frac{\pa m}{\pa k}\pa_{m}\(\frac{1}{mn} \log Z_{m,n}\)+\frac{\pa n}{\pa k}\pa_{n}\(\frac{1}{m n} \log Z_{m,n}\)\) \rt|_{k=1}\\
    &= -\frac{2}{5}\log Z_{2,1} -4\(-\frac{1}{10}\frac{A_2(\gamma_C)}{4G}-\frac{1}{4}\frac{2 EW^{(2)}}{4G}\)\\
    &= 2 \frac{EW^{(2)}}{4G} + \frac{2}{5}\tilde{S}_2(AB)+\frac{2}{5} S_2(AB),
\end{align}
where $EW^{(2)}$ is the entanglement wedge cross section that computes the $(2,1)$ R\'enyi reflected entropy/CCNR \cite{Milekhin:2022zsy}.
Moreover, $S_n(AB)$ and $\tilde{S}_n(AB)$ are the R\'enyi and refined R\'enyi entropy for two intervals in the vacuum state.

In order to show that our saddle is better than this, we will show
\begin{equation}
    \tilde{S}_n(AB) \geq \frac{2}{n} \tilde{S}_2(AB)
    \label{eq:inequalnumer}
\end{equation}
within the range $n\in [1,2]$. This can be reliably computed using CFT methods, which we use to demonstrate the validity of the inequality \Eqref{eq:inequalnumer} numerically in Figure \ref{fig:saddle_compare}. From \Eqref{eq:inequalnumer}, we have
\begin{align}
    S_2(AB) &= 2\int_1^2\frac{\tilde{S}_{n'}(AB)}{n'^2} dn'\\
    & \geq 4\tilde{S}_{2}(AB)\int_1^2 \frac{1}{n'^3}dn' = \frac{3}{2}\tilde{S}_2(AB).
\end{align}
From this it follows that,
\begin{equation}
    \tilde{\mathcal{E}}^{(2)}_{CFT} \geq 2 \frac{EW^{(2)}}{4G} + \tilde{S}_2(AB),
\end{equation}
but it is easy to see geometrically that our proposed saddle is strictly better than this by smoothing corners.

\begin{figure}
    \centering
    \includegraphics{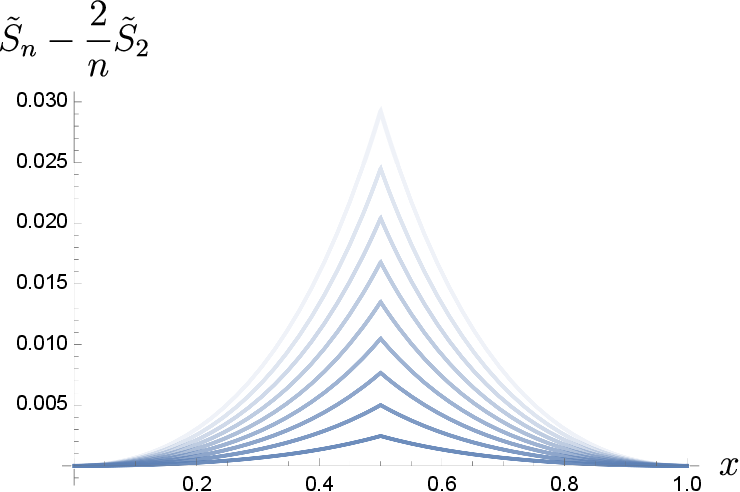}
    \caption{The difference in refined R\'enyi entropies (normalized by $c$) for $n$ from $1.1$ (bottom line) to $1.9$ (top line) in steps of $0.1$. These are approximated by the vacuum conformal block using Zamolodchikov's recursion relation. Clearly, \eqref{eq:inequalnumer} is satisfied. We note that in the adjacent ($x\rightarrow 1$) and distant ($x\rightarrow 0$) limits, the difference disappears.}
    \label{fig:saddle_compare}
\end{figure}

It is also useful to note that the geometric proposal of Ref.~\cite{2019PhRvL.123m1603K} continues to disagree with our proposal even in the adjacent interval limit. For adjacent intervals of lengths $l_1$ and $l_2$, the even moments are given by
\begin{align}
      \Tr\left( \rho_{AB}^{T_B}\right)^{2k} = \left \langle \sigma_{2k}(-l_1) \sigma_{2k}^{-2}(0)  \sigma_{2k}(l_2)\right \rangle.
\end{align}
The scaling of the three-point function is fixed by conformal invariance, such that
\begin{align}
\Tr\left( \rho_{AB}^{T_B}\right)^{2k} = \frac{C_{\sigma_{2k}\sigma_{2k}\sigma_{2k}^{-2}}}{(l_1l_2)^{\Delta_{2k}^{(-2)}}(l_1 + l_2)^{2\Delta_{2k} - \Delta_{2k}^{(-2)}}},
\end{align}
where $\Delta_{2k}^{(-2)}$ denotes the conformal dimension of $\sigma_{2k}^{-2}$.
Our proposal and that of Ref.~\cite{2019PhRvL.123m1603K} both reproduce the correct scaling, though they disagree on the value of $C_{\sigma_{2k}\sigma_{2k}\sigma_{2k}^{-2}}$. For the logarithmic negativity, this leads to a relative constant shift between the two proposals that can be tested from the CFT.

\section{Discussion}
\label{sec:disc}

\subsection{Odd/Transposed entropy} 
\label{sub:odd}

We have mainly focused on the even moments of the partial transpose due to their relation with the negativity. The odd moments have also been proposed to be useful as an entanglement measure. Namely, Ref.~\cite{2019PhRvL.122n1601T} introduced the odd entropy, later called the partially transposed entropy in Ref.~\cite{2021JHEP...06..024D}. It is defined as
\begin{align}
    S^{T_B}(A:B) = -\sum_{i}\lambda^{(T)}_i \log |\lambda^{(T)}_i|.
\end{align}
This may also be evaluated using a replica trick by analytically continuing the odd moments
\begin{align}
    S^{T_B}(A:B) =  \lim_{k \rightarrow 1} \frac{1}{2(1-k)} \log \Tr \left(\rho_{AB}^{T_B}\right)^{2k - 1}.
\end{align}
Similar to the tension for holographic negativity, Refs.~\cite{2019PhRvL.122n1601T} and \cite{2021JHEP...06..024D} have conflicting proposals for the holographic dual of $S^{T_B}$. Using 2D holographic CFT techniques (nearly identical to the incorrect derivation of Petz R\'enyi mutual information in Appendix~\ref{sub:PRMI}), Ref.~\cite{2019PhRvL.122n1601T} showed that $S^{T_B}$ was equal to the area of the entanglement wedge cross section plus the RT surface, without any backreaction for either surface. In contrast, Ref.~\cite{2021JHEP...06..024D} showed that for fixed-area states, $S^{T_B}$ was equal to half the mutual information plus the area of the RT surface. 

The proposal of Ref.~\cite{2019PhRvL.122n1601T} again assumes the full $\mathbb{Z}_{2k-1}$ symmetry for calculating $\Tr \left(\rho_{AB}^{T_B}\right)^{2k - 1}$.
For the case of adjacent intervals in the vacuum state of a 2D CFT, this leads to a bulk configuration whose quotient space again has three intersecting branes as shown in \figref{fig:triway}.
The only difference from \figref{fig:triway} is that now the number of copies being glued together is $2k-1$ and the permutations $\{X,X^{-1},e\}$ are correspondingly the cyclic, anti-cyclic, and identity permutations on $2k-1$ elements.
By the same argument made in \secref{sec:intervals}, we can just look at the topology of the Riemann surface computing the boundary partition function.
For the case of the odd entropy, the genus is $k-1$ and thus, we again see that for $k\geq 2$, this space time is not a smooth manifold.\footnote{Similar arguments can be used to rule out replica symmetric saddles for the multi-entropy discussed in Refs.~\cite{Gadde:2022cqi,Gadde:2023zzj}. This argument however does not work for the reflected entropy since there the topology around the intersection point ends up being a sphere.}
Our study of brane intersections in Appendix~\ref{app:brane} makes this precise more generally.

More generally, one may hope to use continuity bounds on the odd entropy in order to find its holographic dual, given that we already know the answer for fixed-area states. The fact that the computation of odd entropy involves tensionless branes makes it promising for it to have a continuity bound similar to other such quantities like entanglement entropy \cite{cmp/1103859037} and reflected entropy \cite{2020JHEP...04..208A}. However, the partially transposed density matrix does not satisfy good continuity properties as can be checked in simple examples.\footnote{We thank Isaac Kim for discussions related to this.} 

\subsection{Replica symmetry restoration for RTNs}
\label{sub:nmRTNdisc}

Random tensor networks have been demonstrated to be very useful models of holographic states \cite{2016JHEP...11..009H}, clarifying various information theoretic aspects of the holographic mapping. In standard random tensor networks, the links in the network are taken to be maximally entangled, which causes the entanglement spectrum for states to be ``flat,'' i.e., all R\'enyi entropies are equal. This is unlike general states in conformal field theory, where the spectrum is far from being flat \cite{2008PhRvA..78c2329C}. As was noted in \cite{2016JHEP...11..009H}, this can be implemented in random tensor networks by having the links in the network be non-maximally entangled. In Appendix \ref{sec:RTN}, we demonstrate that when sufficiently non-maximally entangled states are taken for the links in the network, the replica symmetric saddle becomes dominant over the replica symmetry breaking saddles. We comment on why the conclusion in random tensor networks is so different from full quantum gravity that we have discussed in the main text. We expect this to be useful in the pursuit of better tensor network models of AdS/CFT.

\subsection{Implications for holography}

An important future direction is to explore what the quantum information theoretic implications of this holographic dual of the negativity are.
It is well known that the Ryu-Takayanagi formula led to a much better understanding of the bulk-boundary dictionary in AdS/CFT.
It remains for us to understand what mileage we can gain from our holographic prescription for the negativity.

For instance, Refs.~\cite{2021arXiv211200020V,2022PhRvL.129f1602V} argued that there are instances in holography where the negativity can be large while the mutual information is small.
They interpreted such situations as consisting of bound entanglement between the two parties, which is not distillable.
Using our holographic prescription for the negativity, it is easy to see that such cases arise quite generally in the presence of entanglement phase transitions.
The negativity, and the ERNs more generally, are sensitive to the doubled state which corresponds to a saddle computing the second R\'enyi entropy.
In general, the phase transition in R\'enyi entropy can happen at different locations in the parameter space depending on the R\'enyi parameter.
This would give rise to situations where the negativity can be large while the mutual information is small, the case of an evaporating black hole being a particular example.

\acknowledgments
We thank Chris Akers, Ahmed Almheiri, Eugenia Colafranceschi, Tom Faulkner, Abhijit Gadde, Tom Hartman, Matt Headrick, Isaac Kim, Yuya Kusuki, Simon Lin, Henry Maxfield, Francesco Mele, Sean McBride, Vladimir Narovlansky, Geoff Penington, Xiao-Liang Qi, Mukund Rangamani, Shinsei Ryu, Michael Walter, and Wayne Weng for helpful discussions and comments. JKF specially thanks Yuya Kusuki, Vladimir Narovlansky, and Shinsei Ryu for many discussions and previous collaboration on related work. XD is supported in part by the U.S. Department of Energy, Office of Science, Office of High Energy Physics, under Award Number DE-SC0011702 and by funds from the University of California. JKF is supported by the Marvin L.~Goldberger Member Fund at the Institute for Advanced Study and the National Science Foundation under Grant No. PHY-2207584. PR is supported in part by a grant from the Simons Foundation, by funds from UCSB, the Berkeley Center for Theoretical Physics; by the Department of Energy, Office of Science, Office of High Energy Physics under QuantISED Award DE-SC0019380, under contract DE-AC02-05CH11231 and by the National Science Foundation under Award Number 2112880. 
This material is based upon work supported by the Air Force Office of Scientific Research under award number FA9550-19-1-0360.

\appendix

\section{Petz R\'enyi Mutual Information}
\label{sub:PRMI}

To gain further understanding of this misidentification of the dominant channel, it is useful to consider a similar quantity called the Petz R\'enyi mutual information (PRMI), which may be evaluated using a similar replica trick \cite{2022arXiv221101392K,2023arXiv230808600K}. Unlike the usual linear combination of R\'enyi entropies frequently studied in the literature, the PRMI is a well-behaved generalization of the mutual information in that it is never negative and monotonically decreases under quantum channels. This is a consequence of its definition using the Petz R\'enyi relative entropy
\begin{align}
    I_{\alpha}(A,B) := D_{\alpha}(\rho_{AB}||\rho_{A}\otimes \rho_B):= \frac{1}{\alpha-1}\log\Tr\left[ \rho_{AB}^{\alpha} (\rho_{A}\otimes \rho_B)^{1-\alpha}\right].
\end{align}
The usefulness of the quantity for our purposes is that, by definition, it must limit to the standard mutual information in the $\alpha \rightarrow 1$ limit, and we know with confidence, from the Ryu-Takayanagi formula, what the holographic dual for mutual information is (away from phase transitions). 

We now demonstrate that using the same assumptions for holographic correlation functions used in Refs.~\cite{2019PhRvD..99j6014K, 2019PhRvL.123m1603K} for negativity leads to an answer that we know for certain is incorrect. We then conclude that RSB must be incorporated into CFT computations in order to determine the correct answers.

The replica trick for PRMI involves two replica indices
\begin{align}
     I_{\alpha}(A,B) = \lim_{m\rightarrow 1-\alpha}  \frac{1}{\alpha-1}\log \left[\Tr\left(\rho_{AB}^{\alpha} (\rho_{A}\otimes \rho_B)^{m}\right)\right].
     \label{PREE_replica}
\end{align}
The joint moments may be evaluated using twist fields implementing a $g_A$ permutation on region $A$ and a $g_B$ permutation on region $B$, where in cycle notation
\begin{align}
    g_A &= (1,\dots, \alpha, \alpha + 1,\dots, \alpha + m), 
    \quad
    g_B = (1,\dots, \alpha, \alpha + m+1,\dots, \alpha + 2m).
\end{align}
The analogue of the double twist operator is the twist field corresponding to the permutation coming from fusing $g_A^{-1}$ and $g_B$
\begin{align}
    g_{A}^{-1} g_B = (1,\alpha + m, \alpha + m-1,\dots , \alpha+ 1,\alpha+m+1,\alpha + m +2,\dots ,\alpha +2m).
\end{align}
The conformal dimensions are fixed by the cycle structures
\begin{align}
    \Delta := \Delta_{g_A} = \Delta_{g_B} &= \frac{c}{12}\left(\alpha + m -\frac{1}{\alpha + m} \right),
    \quad
    \Delta_{g_A^{-1}g_B} = \frac{c}{12}\left(2m+1 -\frac{1}{2m+1} \right).
\end{align}
For disjoint intervals, the moments are given by
\begin{align}
    \Tr\left(\rho_{AB}^{\alpha} (\rho_{A}\otimes \rho_B)^{m}\right) = \left \langle \sigma_{g_A}(x_1) \sigma_{g_A}^{-1}(x_2) \sigma_{g_B}(x_3) \sigma_{g_B}^{-1}(x_4)\right \rangle.
\end{align}
For close intervals, we expand in the t-channel
\begin{align}
    \Tr\left(\rho_{AB}^{\alpha} (\rho_{A}\otimes \rho_B)^{m}\right)= \sum_p |C_{ABp}|^2 \mathcal{F}_p(1-x)\bar{\mathcal{F}}_p(1-\bar{x}), \quad 
    x := \frac{(x_2-x_1)(x_4-x_3)}{(x_3-x_1)(x_4-x_2)}.
\end{align}
We assume that at large $c$, we only need to keep the $g_A^{-1}g_B$ twist field in the sum. Unlike the case of negativity, we note that as $m\rightarrow 1-\alpha$ and $\alpha \rightarrow 1$, all operators become light. In such a limit, the Virasoro conformal blocks are known analytically at large $c$ \cite{2015JHEP...11..200F}, giving
\begin{align}
    \lim_{\alpha \rightarrow 1}I_{\alpha}(A,B) = \frac{c}{3}\log \left(\frac{1 + \sqrt{x}}{1-\sqrt{x}} \right)+O(1),
\end{align}
where the additive constant comes from the OPE coefficient and is not important for our purposes. This is proportional to the area of the entanglement wedge cross section in the vacuum \cite{2018NatPh..14..573U}, which is very different from the known answer for mutual information
\begin{align}
    I(A,B) = \frac{c}{3}\log \left(\frac{x}{1-x}\right).
\end{align}
Clearly, our assumption regarding the dominant conformal block was incorrect.

For the PRMI, there is also a clear analogue of the RSB saddle in the bulk. The RSB permutation that lies simultaneously on geodesics between $g_A$ and $\mathbbm{1}$ and between $g_B$ and $\mathbbm{1}$ with the most residual symmetry is
\begin{align}
    X = (1,\dots, \alpha).
\end{align}
The compositions of this permutation with $g_A$ and $g_B$ are
\begin{align}
    X^{-1}g_A &= (1,\alpha + 1,\dots, \alpha + m),
    \quad
    X^{-1}g_B =  (1,\alpha + m+ 1,\dots, \alpha +2 m).
\end{align}
It is clear that we need a general CFT prescription for evaluating RSB saddles to obtain the correct answer. The naive argument for single-block dominance in the conformal block decomposition is insufficient.

\section{Random Tensor Networks vs.\ Gravity}
\label{sec:RTN}
 
In this appendix, we consider random tensor networks, simple toy models of holographic duality that are remarkably effective in modeling the information theoretic aspects of AdS/CFT \cite{2010JPhA...43A5303C, 2016JHEP...11..009H}. The negativity has been analyzed in random tensor networks where replica symmetry breaking saddles provided the dominant contributions \cite{2021JHEP...06..024D, 2021PRXQ....2c0347S,2022JHEP...02..076K}. A key feature of these tensor networks was that their link states were maximally entangled, such that the entanglement spectra were approximately flat. This is unlike the entanglement spectra in general holographic states, and it has been suggested that a better model of holographic states can be made by modifying the link states to be sub-maximally entangled such that their spectra are not flat \cite{2016JHEP...11..009H, 2022arXiv220610482C}. We consider this modification and find that for sufficiently non-flat link spectra, replica symmetry is restored. We first review the construction of random tensor networks, explain the mechanism for replica symmetry restoration, and then comment on why this conclusion differs from that in a full theory of gravity.  

\subsection{Random tensor networks with non-flat link spectra}

A tensor network is defined on a graph with vertices $V$, and edges $E$ connecting pairs of vertices. For each vertex $x \in V$, we assign a rank-$k$ tensor, $T^{(x)}_{i_1 \dots i_k}$, where $k$ is the number of edges connected to $x$. Each tensor defines a state
\begin{align}
    \ket{T_x} =\sum_{i_1 \dots i_k} T^{(x)}_{i_1 \dots i_k}\ket{i_1}\dots \ket{i_k},
\end{align}
where the states on the right-hand side are basis vectors. To each edge $\{xy \} \in E$, we define a rank-$2$ tensor $E^{(xy)}_{ij}$ with the corresponding state
\begin{align}
     \ket{xy} = \sum_{i,j} E^{(xy)}_{ij}\ket{i}_x\ket{j}_y.
\end{align}
Frequently, these are taken to be maximally entangled (up to normalization), i.e., $E^{(xy)}_{ij} = \delta_{ij}$. The total tensor network state is then defined as
\begin{align}
    \ket{T} = \left(\bigotimes_{\{ xy\} \in E}\bra{xy} \right)\left(\bigotimes_{x \in V} \ket{T_x} \right).
\end{align}

In a random tensor network, the $\ket{T_x}$'s are drawn from the uniform (Haar) measure on a $D_x$-dimensional Hilbert space, where $D_x$ is the local Hilbert space dimension at a given vertex $x$. The (unnormalized) moments are then given by
\begin{align}
    \overline{\ket{T_x} \bra{T_x}^{\otimes k}}  =  {\sum_{\tau \in Sym_k}g_{\tau}}
\end{align}
where $g_{\tau}$ is the matrix representation of permutation $\tau$ in the symmetric group $Sym_k$. In order to compute negativities, we will need the moments of the partially transposed density matrix. These may be evaluated using correlation functions of twist operators, i.e., cyclic ($\Sigma$) and anti-cyclic ($\Sigma^{-1}$) permutations. For example,
\begin{align}
    \overline{\Tr\left[\left( \rho_{AB}^{T_B}\right)^k\right]} = \frac{\overline{\bra{T}^{\otimes k} \Sigma_A \Sigma^{-1}_B \ket{T}^{\otimes k}}}{\overline{\bra{T}\ket{T}^{ k}}} = {\sum_{\{ g_x\}} e^{-\mathcal{A}[\{g_x \}]}}.
\end{align}
In the second equality, the correlation function is reinterpreted as the partition function for a classical statistical mechanics model involving spins valued in $Sym_k$ located at each tensor. $\{ g_x\}$ represents the set of all allowed spin configurations obeying the boundary conditions set by the twist operators. That is, we set the boundary condition to the cyclic permutation $X$ on subregion $A$, the anti-cyclic permutation $X^{-1}$ on $B$, and the identity permutation $e$ on $C$. When all bond dimensions are taken to be large, the spin model will be in its ferromagnetic phase such that the dominant contributions to the partition function are given by simple domain wall configurations between groups of tensors that are all aligned.

The spin model action is given by
\begin{align}
    \mathcal{A}[\{g_x \}] = \sum_{\{ xy\} \in E}J(g_x^{-1}g_y), \quad J(h) = -\log\Tr\left[h\rho_e^{\otimes k} \right],
\end{align}
where $\rho_e$ is the density matrix for the link states, i.e., $\rho^{(xy)}_e  = \Tr_y \ket{xy}\bra{xy}$.
If permutation $h$ contains $C(h)$ cycles of lengths $k_1,\dots ,k_{C(h)}$, then
\begin{align}
    J(h) = \sum_{i = 1}^{C(h)}(k_i -1)S_{k_i}(\rho_e).
\end{align}
In the simplifying case where the links are maximally entangled, all R\'enyi entropies are the same, so
\begin{align}
    J(h) = \left(k- C(h)\right)\log D = d(e,h) \log D,
    \label{cost}
\end{align}
where $D$ is the dimension of the link state $\rho_e$, and $d(g_1,g_2)=k-C(g_1 g_2^{-1})$ is the Cayley distance metric on $Sym_k$.

\subsection{Replica symmetry breaking}

To warm up, let us first consider a tensor ``network'' consisting of a single random tensor with maximally entangled links (the topic of Ref.~\cite{2021PRXQ....2c0347S}). There is only a single spin to sum over in the partition function, i.e.,
\begin{align}
    \overline{\Tr\left[\left( \rho_{AB}^{T_B}\right)^k\right]} = \sum_{h \in S_k} D_A^{C\left(X^{-1}h\right) - k}D_B^{C\left(X h\right) - k}D_C^{C(h)-k},
\end{align}
where the subscripts indicate the different bond dimensions on different links.
The relevant parameter regime for holography is $1/D_C \ll D_A/D_B \ll D_C$. Therefore, to maximize the exponents, we would like to find an $h$ that simultaneously maximizes $C\left(X^{-1}h\right) + C(h)$ and $C\left(X h\right) + C(h) $. This means that $h$ is on the intersection of a geodesic between $X$ and $e$ and a geodesic between $X^{-1}$ and $e$, as measured by the Cayley metric of $Sym_k$. These are given by non-crossing permutations consisting of only one-cycles and two-cycles \cite{2021PRXQ....2c0347S,2021JHEP...06..024D}. 
Furthermore, in the regime where there is significant entanglement between $A$ and $B$ ($D_A D_B/D_C \gg 1$), the number of two-cycles is maximized, such that there is no one-cycle for even $k$ and a single one-cycle for odd $k$. We denote this set of non-crossing pairings as $NC_2$.
Including the contributions of all such elements and computing the spectrum of $\rho_{AB}^{T_B}$ via the resolvent method, one obtains a semicircle distribution centered at $(D_AD_B)^{-1}$ and with radius $2(D_AD_BD_C)^{-1/2}$. At leading order, the logarithmic negativity is then found to equal half of the mutual information
\begin{align}
    \mathcal{E}(A:B) = \frac{1}{2}I(A:B) = \frac{1}{2}\log \frac{D_A D_B}{D_C}.
\end{align}

Considering more general tensor networks with more than one tensor, Ref.~\cite{2021JHEP...06..024D} further showed that the $NC_2$ permutations are the only relevant ones for a large class of RTNs with maximally entangled links.\footnote{For random tensor networks that exhibit negativity spectra different from the semicircle distribution, see \cite{2022JHEP...02..076K}.
} The twist operators that set the boundary conditions in the spin model fix the domain wall structure at the boundary between $X$ or $X^{-1}$ and the identity $e$. One would naively think that these domain walls extend into the bulk as in the left figure of \figref{splitting_fig}. However, as the domain wall ``tension'' for maximally entangled links is given by \Eqref{cost}, the domain wall between the $X$ and $X^{-1}$ domains can split creating a new domain filled in by some permutation $h$, without incurring any energy cost so long as $h$ lies on the intersection of pairwise geodesics between $X$, $X^{-1}$ and $e$. Once the domain walls split, they can relax into their minimal area positions in order to minimize the global energy cost (right figure of \figref{splitting_fig}). The final dominant spin configurations in the partition function consist of a large region in the center filled in by some $\tau \in NC_2$. The calculation thus reduces to that of the single-tensor network with $D_A, D_B, D_C$ replaced by the product of the dimensions of the bonds on the corresponding domain walls. Thus, it is clear that the negativity again equals half of the mutual information.

\subsection{Replica symmetry restoration}

The simplest RTN model that realizes the restoration of replica symmetry by virtue of a non-flat spectrum comprises two random tensors, which we call the 2TN model. 
\begin{align}
    \input{2TN}
\end{align}
This models a generic situation in holography where the internal bond plays the role of the entanglement wedge cross section. This model has been studied in detail in the case of flat link spectra \cite{2022JHEP...02..076K}. We specialize to a particular spectrum motivated by the single interval R\'enyi entropy in 2D CFT and analyze the problem using the resolvent method.

Consider the 2TN model with all link spectra following 2D CFT behavior $S_n\propto \frac{n+1}{2n}$. Explicitly, we may take the links to be
\begin{align}
    \ket{x y} = \int_0^{\lambda_{max}} d\lambda \rho(\lambda) \sqrt{\lambda }\ket{\lambda}_x\ket{\lambda}_y, \quad \lambda_{max} = e^{-S_{vN}/2}
\end{align}
where the density of states is \cite{2008PhRvA..78c2329C}
\begin{align}
    \rho(\lambda) = \delta(\lambda_{max}-\lambda)+ \Theta(\lambda_{max}-\lambda)\frac{{S_{vN}}{}}{\lambda\sqrt{{2S_{vN}}{}\log (\lambda_{max}/\lambda)}}I_1\left(\sqrt{2{S_{vN}}{}\log (\lambda_{max}/\lambda)}\right).
\end{align}
The above spectrum ensures the R\'enyi entropies for single intervals agree with the answer obtained in the vacuum state of a CFT.

In this situation, it is straightforward to check that the RSB saddle is subdominant to the RS saddle. We are then left with comparing the connected and disconnected RS saddles. We have
\begin{align}
     \mathcal{A}^{(n)}_{disc}-\mathcal{A}^{(n)}_{conn} = \begin{cases}
     (n-1)\left({S_{vN,A}+S_{vN,B}-\frac{n+1}{2n}S_{vN,C}}\right) -(n-2){S_{vN,EW}}
     , & n \in \mathbb{Z}_{even}
     \\
          (n-1)\left({S_{vN,A}+S_{vN,B}-\frac{n+1}{2n}S_{vN,C}}-S_{vN,EW}\right) 
     , & n \in \mathbb{Z}_{odd}
     \end{cases}.
\end{align}
In the regime where the entanglement wedge of $AB$ is connected, $S_{vN,A}+S_{vN,B} > S_{vN,C}$, the difference is always positive, so we can safely ignore the disconnected saddle.

To evaluate the negativity, we first find the full negativity spectrum using the resolvent method. The negativity resolvent is defined as
\begin{align}
\label{resolvent_def}
    R(z) = \Tr\left(\frac{\rho_{AB}^{T_B}}{z-\rho_{AB}^{T_B}}\right).
\end{align}
The negativity spectrum is given by the discontinuity over the real axis
\begin{align}
    \lambda \rho_{\mathcal{E}}(\lambda) = -\frac{1}{\pi} \Im R(\lambda + i\epsilon) \Big |_{\epsilon \rightarrow 0^+}.
\end{align}
Following similar calculations to \cite{2008PhRvA..78c2329C}, we obtain
\begin{align}
   \rho_{\mathcal{E}}(\lambda) &=   \frac{\delta_{\[0,\lambda_{max} \]}
   }{2 \lambda  \sqrt{\log (\lambda_{max}/\lambda)}}
   \Bigg(\sqrt{\frac{S_{vN,C}+4S_{vN,EW}}{2}} I_1\left(2 \sqrt{\frac{S_{vN,C}+4S_{vN,EW}}{2}} \sqrt{\log
   (\lambda_{max}/\lambda)}\right)
   \nonumber
   \\
   & +\sqrt{\frac{S_{vN,C}+S_{vN,EW}}{2}} I_1\left(2 \sqrt{\frac{S_{vN,C}+S_{vN,EW}}{2}}
   \sqrt{\log (\lambda_{max}/\lambda)}\right)\Bigg)
   \nonumber
   \\
   &-\frac{\delta_{\[ -\lambda_{max},0 \]}}{2 \lambda  \sqrt{\log (-\lambda_{max}/\lambda)}}
   \Bigg(\sqrt{\frac{S_{vN,C}+4S_{vN,EW}}{2}} I_1\left(2 \sqrt{\frac{S_{vN,C}+4S_{vN,EW}}{2}} \sqrt{\log
   (-\lambda_{max}/\lambda)}\right)
   \nonumber
   \\
   & -\sqrt{\frac{S_{vN,C}+S_{vN,EW}}{2}} I_1\left(2 \sqrt{\frac{S_{vN,C}+S_{vN,EW}}{2}}
   \sqrt{\log (-\lambda_{max}/\lambda)}\right)\Bigg) + \delta\left(\lambda -  \lambda_{max}\right),
   \label{neg_spec_final}
\end{align}
where $\lambda_{max}:= e^{-\frac{S_{vN,C}+S_{vN,EW}}{2}}$.
The above spectrum reproduces all the integer moments of $\rho_{AB}^{T_B}$. One may furthermore evaluate the RRNs and logarithmic negativity directly from the spectrum to find
\begin{align}
    \tilde{\mathcal{E}}^{(n)}(A:B)
    &= \log\left( \int d\lambda \rho_{\mathcal{E}}(\lambda) \left|\lambda \right|^n\right) - n \frac{\int d\lambda \rho_{\mathcal{E}}(\lambda) \left|\lambda\right|^n\log \left|\lambda\right|}{\int d\lambda \rho_{\mathcal{E}}(\lambda)\left| \lambda \right|^n}
    =\frac{S_{vN,C}+4S_{vN,EW}}{n},
    \\
    {\mathcal{E}}(A:B) &= 
    \log \left(\int d\lambda \rho_{\mathcal{E}}(\lambda) |\lambda| \right)= 
    \frac{3}{2}S_{vN,EW},
\end{align}
in agreement with the naive analytic continuation.

\subsection{Comparison to gravity}

We have explicitly shown that the replica symmetric saddle is the dominant contribution when there are sufficiently non-flat link spectra. It is instructive to analyze why this conclusion was distinct from gravity.

Consider computing the RRN for even integer $2k$ using the gravitational path integral for two intervals in vacuum AdS. For the candidate RSB saddle, we focus on the permutation $(1,2)(3,4)\dots (2k-1,2k)$ because it retains a $\mathbb{Z}_{k}$ replica symmetry that cyclically permutes the pairs of copies.\footnote{The degeneracy between different choices of RSB saddle breaks once we move away from the fixed-area limit. Perturbatively, for nfRTNs, it can be shown that this choice of saddle is indeed the best RSB candidate. However, it remains an open interesting question whether this is true in gravity. For our analysis, we will assume this is the case.} It is convenient to quotient the bulk by this symmetry, giving a bulk geometry, $\hat{\mathcal{B}}_2^{(k)}$, whose asymptotic boundary is a two-fold cover of the original boundary, branched over $A\cup B$. In the quotient space, there are conical defects at the fixed points of the quotient with opening angle $\frac{2\pi}{k}$. These are homologous to subregions $A$ and $B^*$ as shown in \figref{fig:RRN}. At $k = 1$, the defects disappear and the geometry is smooth. If this saddle dominates, the RRN is given by 
\begin{align}\label{eq:RRNgrav}
    \tilde{\mathcal{E}}^{(2k)} = k^2 \partial_k I\[\hat{\mathcal{B}}_2^{(k)}\] = \frac{\text{Area}\left(\gamma_{A}^{(2\pi/k)} \cup \gamma_{B^*}^{(2\pi/k)}: \hat{\mathcal{B}}_2^{(k)}\right)}{4G},
\end{align}
where we remind the reader that $\g^{(2\pi/k)}$ means a conical defect of opening angle $\frac{2\pi}{k}$.
At $k=1$, \Eqref{eq:RRNgrav} gives the area of the surface $\gamma_{A}\cup \gamma_{B^*}$ in $\mathcal{B}_2$. $\mathcal{B}_2$ is locally the original single-copy geometry with the additional backreaction of R\'enyi-2 branes of tension $\frac{1}{8 G}$ located at the RT surface of $AB$ (which we will call $\gamma_C$). Note that at $k=1$ the backreaction from the surfaces $\gamma_{A}, \gamma_{B^*}$ vanishes, and thus, it does not matter whether we compute the area of $\gamma_{A}$ or $\gamma_{A^*}$ since there is a symmetry relating them. 

\begin{figure}
    \centering
    \includegraphics[scale=0.5]{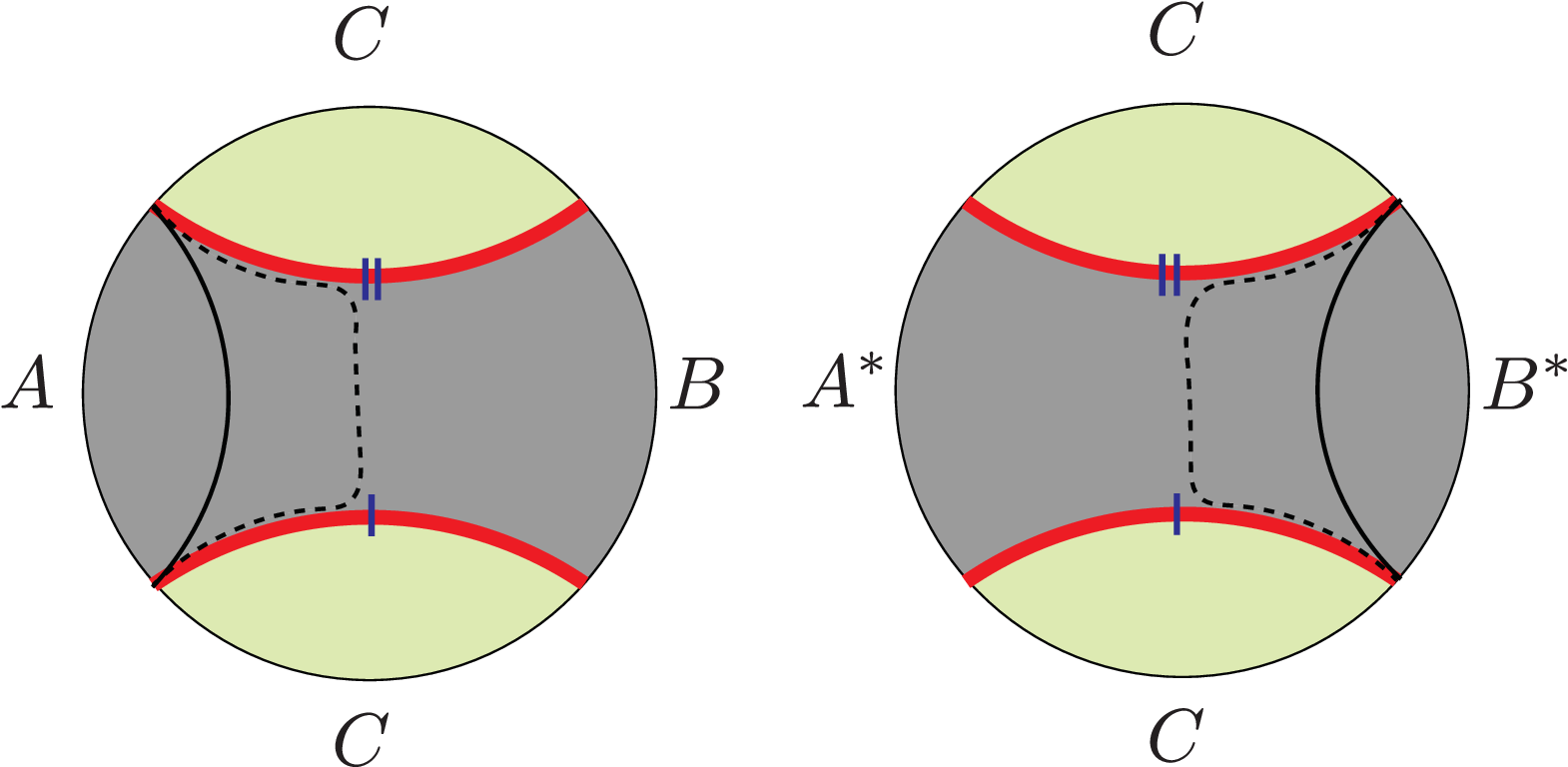}
    \caption{The RRN $\tilde{\mathcal{E}}^{(2)}$ is computed by the sum of minimal surfaces (solid, black) homologous to subregions $A$ and $B^*$ in a gravitational solution $\mathcal{B}_2$ which has the topology of a double-cover of the original spacetime branched over the ``cosmic branes'' (red) anchored on subregion $AB$. The analogous answer obtained in nfRTNs is given by the dashed line, which is non-minimal.}
    \label{fig:RRN}
\end{figure}

There is a new effect here not seen in the nfRTN. We understand this effect in gravity as follows: in $\hat{\mathcal{B}}_2^{(k)}$, there is a backreaction effect due to the branching at $\gamma_{C}$ that has the effect of a R\'enyi-2 brane\footnote{Strictly speaking, we have a R\'enyi-2 brane only after taking a $\Z_2$ quotient of $\hat{\mathcal{B}}_2^{(k)}$, but the other R\'enyi-$k$ branes break the $\Z_2$ symmetry. Nonetheless, we will refer to the backreaction effect in $\hat{\mathcal{B}}_2$ as that of a R\'enyi-2 brane.} with tension $\frac{1}{8 G}$. There is also a backreaction effect due to the R\'enyi-$k$ branes at $\gamma_{A}$ and $\gamma_{B^*}$ with tensions $\frac{k-1}{4 kG}$. Near the asymptotic boundary, these branes converge, and naively adding their tensions gives $\frac{3k-2}{8 k G}$, which (for $k>1$) is larger than the $\frac{2k-1}{8 k G}$ tension of a R\'enyi-$2k$ brane that one would have from the replica symmetric solution if taking the full $\mathbb{Z}_{2k}$ quotient. This larger tension would naively suggest that the RSB saddle is always subdominant due to this IR divergence, just as in the nfRTN. However, this naive argument has crucially neglected the fact that in gravity, the branes will backreact \textit{on each other}, which will be in just the right way to cancel this effect. It is this mutual backreaction that is not captured by the nfRTN, causing the RS saddle to dominate in the nfRTN.

We mention that a similar effect has been observed previously in the literature of nfRTN \cite{2016JHEP...11..009H} in the simpler case of the R\'enyi entropies of disjoint intervals. Because the R\'enyi entropies are holographically computed from the areas of tensionful branes \cite{2016NatCo...712472D}, the additional gravitational action due to even very distant branes is different from the sum of additional actions of the solutions with just one brane, owing to the mutual backreaction between the branes. This implies that the R\'enyi mutual information is never zero, even at leading order in $G$ \cite{2016NatCo...712472D,2010PhRvD..82l6010H}. While it is caused by a similar mechanism, this well-known example is a milder critique of nfRTN as models of holography than the negativity case that has been the topic of this paper because at least, the nfRTN faithfully captures the correct bulk saddle topology for the R\'enyi entropy.

\section{Ruling Out Intersecting Branes}\label{app:brane}

In this appendix, we prove that geometries $\hat{g}$ with certain intersecting branes or branes ending on other branes cannot be obtained from a $\Z_k$ quotient of a smooth parent geometry $g$. Therefore, these geometries are not saddles of the gravitational path integral and will not contribute to calculations such as for the negativity.  

\begin{theorem}
\label{thm1}
    Let $\hat{g}$ be a geometry with intersecting branes or a brane ending on another
brane. If all involved branes are codimension-2 and at least one has conical opening angle $2\pi/k$, then $\hat{g}$
cannot be obtained by a $\Z_k$ quotient of a smooth geometry $g$.
\end{theorem}

In order to prove this theorem, we first prove the following lemma.

\begin{nlemma}\label{fixedpoint}
    Under the assumptions of Theorem \ref{thm1} and further supposing that $\hat{g}$ is the quotient of a smooth geometry $g$ by a $\Z_k$ isometry generated by $r$, any regular point on any brane with conical opening angle $2\pi/k$ in $\hat{g}$ must be a fixed point of $r$. Here, a regular point on a brane is defined as a point at which the union of all the branes is locally a smooth manifold, thus omitting intersection points.
\end{nlemma}

\begin{proof}
    Let $p$ be a regular point on a brane. We aim to rule out the possibility that $p$ is a fixed point of a power of $r$, but not $r$ itself. In a sufficiently small neighborhood of $p$, $g$ is approximately a Euclidean space $E_D$ (where $D$ is the dimension). The isometry group of the Euclidean space, $ISO(E_D)$, is generated by translations, rotations, and reflections. It is well known that the set of fixed points, $\operatorname{Fix}(s)$, of any $s \in ISO(E_D)$ in $E_D$ is an affine subspace of $E_D$.

    By definition, the union of all the branes is the fixed point set $\bigcup_{m=1}^{k-1} \operatorname{Fix}(r^m)$, i.e., the union of the fixed point sets of all non-identity elements of the group $\Z_k$. In a sufficiently small neighborhood of a regular point $p$ on a codimension-2 brane, the brane is approximately a codimension-2 plane, and each $\operatorname{Fix}(r^m)$ can be viewed as an affine subspace of $E_D$: either $r^m$ acts within the neighborhood and thus can be identified with an element of $ISO(E_D)$, or $r^m$ does not act within the neighborhood\footnote{This happens in the examples studied in Section 3 of Ref. \cite{Haehl:2014zoa}.} and thus $\operatorname{Fix}(r^m)$ is empty in the neighborhood. Since each $\operatorname{Fix}(r^m)$ is an affine space, their union can only be a codimension-2 plane if $\operatorname{Fix}(r^m)$ is the full plane for some $m$ and all other $\operatorname{Fix}(r^m)$ are subspaces of the plane. Let $m_0$ be the smallest $m \in \{1,2,\dots,k-1 \}$ such that $\operatorname{Fix}(r^{m})$ is the codimension-2 plane. If $m_0 = 1$, every point on the codimension-2 plane, including $p$, is a fixed point of $r$, and this shows what we wanted to prove.
    
    If $m_0 > 1$, we now derive a contradiction. For $1\leq m \leq m_0-1$, by assumption, $\operatorname{Fix}(r^m)$ is a proper subset of the plane, and since $\operatorname{Fix}(r^m)$ must still be affine, it must be of higher codimension than two. Therefore, we can find a point $q$ in the neighborhood of $p$ that is in $\operatorname{Fix}(r^{m_0})$ but not in $\operatorname{Fix}(r^m)$ for any $1\leq m \leq m_0-1$. This implies that $q$ is a fixed point of any element of $\langle r^{m_0}\rangle$ (the group generated by $r^{m_0}$) but not a fixed point of any other element of $\Z_k$.
    We then use the fact (which we prove in Lemma \ref{rotation} below) that if $s\in ISO(E_D)$ generates a group of finite order $n$ and $\operatorname{Fix}(s)$ is a codimension-2 plane, then $s$ must be a rotation of order $n$ in some 2-plane. Thus, $r^{m_0}$ must be a rotation of order $|\langle r^{m_0}\rangle|$ (the size of $\langle r^{m_0}\rangle$). Since $q$ is a fixed point of $r^{m_0}$ and its powers but not other elements of $\Z_k$, the conical opening angle at $q$ in $\hat g$ must be $2\pi/|\langle r^{m_0}\rangle|$, which contradicts our assumption that the conical opening angle is $2\pi/k$, completing the proof of the lemma.
\end{proof}

In the proof above, we promised to prove the following lemma.

\begin{nlemma}\label{rotation}
    If $s\in ISO(E_D)$ generates a group $\<s\>$ of finite order $n$, and $\operatorname{Fix}(s)$ is a codimension-2 plane, then $s$ must be a rotation of order $n$ in some 2-plane.
\end{nlemma}

\begin{proof}
$s$ generally acts as
\be
x' = M x + v, \qqu \forall x \in E_D,
\ee
where $v$ is a $D$-vector, and $M \in O(D)$. Choose coordinates so that $\operatorname{Fix}(s)$ is the codimension-2 plane given by $x^1=x^2=0$. In other words, as long as $x^1=x^2=0$, we have
\be
x = M x + v.
\ee
Writing this in a block form and, in particular, writing $x=(0,x_2)$ where $0$ denotes a 2-vector specifying the first two components, we find
\be
\begin{pmatrix}
0\\
x_2
\end{pmatrix}=
\begin{pmatrix}
M_{11} & M_{12}\\
M_{21} & M_{22}
\end{pmatrix}
\begin{pmatrix}
0\\
x_2
\end{pmatrix}+
\begin{pmatrix}
v_1\\
v_2
\end{pmatrix}=
\begin{pmatrix}
M_{12} x_2+v_1\\
M_{22} x_2+v_2
\end{pmatrix}.
\ee
In particular, we have
\be
0 = M_{12} x_2+v_1, \qu \forall x_2.
\ee
Setting $x_2=0$, we find $v_1=0$. Thus $0 = M_{12} x_2,\, \forall x_2$. Therefore $M_{12} = 0$.
Similarly, we have
\be
x_2 = M_{22} x_2+v_2, \qu \forall x_2.
\ee
Setting $x_2=0$, we find $v_2=0$. Thus $x_2 = M_{22} x_2,\, \forall x_2$. Therefore $M_{22} = 1$.
Now, since $M\in O(D)$, we find
\be
1= M^T M =
\begin{pmatrix}
M_{11}^T & M_{21}^T\\
0 & 1
\end{pmatrix}
\begin{pmatrix}
M_{11} & 0\\
M_{21} & 1
\end{pmatrix}=
\begin{pmatrix}
M_{11}^T M_{11} + M_{21}^T M_{21} & M_{21}^T\\
M_{21} & 1
\end{pmatrix}.
\ee
Thus, we find $M_{21}=0$ and $M_{11}^T M_{11}=1$. Using $M M^T=1$, we find $M_{11} M_{11}^T=1$. Thus $M_{11} \in O(2)$. It is well known that any such $M_{11}$ must be either a rotation by some angle $\p$ or a reflection in some direction. But in the case of a reflection, it is clear that the resulting $\operatorname{Fix}(s)$ would be codimension-1 instead of codimension-2. Therefore, $M_{11}$ is a rotation by some angle $\p$, and the action of $s$ can be written as a rotation in the $12$-plane:
\be
x' = \begin{pmatrix}
\cos\p & \sin\p & 0\\
-\sin\p & \cos\p & 0\\
0 & 0 & 1
\end{pmatrix} x,
\ee
where $1$ denotes the identity matrix of dimension $D-2$. Since $\<s\>$ has order $n$, $s$ must be a rotation of order $n$. In other words, $\p = 2\pi m/n$ with $(m,n)=1$.
\end{proof}

    We now prove Theorem \ref{thm1}. Suppose that $\hat{g}$ is the quotient of a smooth geometry $g$ by a $\Z_k$ isometry generated by $r$. Let $p$ be a non-regular point on the branes (e.g., an intersection point or ending point between two branes). In a sufficiently small neighborhood of $p$, let $q_1$ be a regular point on the first brane and $q_2$ a regular point on the second. Without loss of generality, we take the second brane to have conical opening angle $2\pi/k$. By definition, the angle between $\overrightarrow{pq_1}$ and $\overrightarrow{pq_2}$ is not $0$ or $\pi$.
    According to Lemma \ref{fixedpoint}, $q_2$ is a fixed point of $r$. We cannot guarantee the same for $q_1$, but by definition $q_1$ is a fixed point of some $r^m$ with $1\leq m\leq k-1$. Therefore, $q_1$ and $q_2$ both belong to $\operatorname{Fix}(r^m)$, which must be an affine space. Thus, any affine combination $\lambda q_1 + (1-\lambda)q_2$ must also be in $\operatorname{Fix}(r^m)$, which is part of the branes. But this is a contradiction: such an affine combination cannot be on the branes, since the angle between $\overrightarrow{pq_1}$ and $\overrightarrow{pq_2}$ is neither $0$ or $\pi$. Thus, we conclude that such a smooth geometry $g$ does not exist, proving Theorem \ref{thm1}.

\addcontentsline{toc}{section}{References}
\bibliographystyle{JHEP}
\bibliography{main}

\end{document}

%% file: 2TN.tex
\begin{tikzpicture}[scale=1]


\draw[black,thick,rounded corners] (0-7,0+1/2) rectangle (1-7,1+1/2);

\draw[black,thick,rounded corners] (2-7,0+1/2) rectangle (3-7,1+1/2);

\node[black] at (1/2-7,1/2+1/2) {$T_1$};
\node[black] at (2+1/2-7,1/2+1/2) {$T_2$};

\draw[black, thick] (1-7,1/2+1/2) -- (2-7, 1/2+1/2);

\draw[black, thick] (1/2-7,0+1/2) -- (1/2-7, -1/2+1/2);
\draw[black, thick] (5/2-7,0+1/2) -- (5/2-7, -1/2+1/2);
\draw[black, thick] (-1-7,1/2+1/2) -- (0-7, 1/2+1/2);

\draw[black, thick] (3-7,1/2+1/2) -- (4-7, 1/2+1/2);

\node[left] at (-1-7,1/2+1/2) {$A$};

\node[right] at (4-7,1/2+1/2) {$B$};

\end{tikzpicture}